\documentclass[11pt, draftclsnofoot, onecolumn]{IEEEtran}
\usepackage{amsmath,amssymb,amsthm}
\usepackage[hidelinks]{hyperref}
\usepackage{xcolor}

\newtheorem{theorem}{Theorem}
\newtheorem{lemma}{Lemma}
\newtheorem{proposition}{Proposition}

\theoremstyle{definition}
\newtheorem{definition}{Definition}
\theoremstyle{remark}
\newtheorem{remark}{Remark}

\newcommand{\G}{\bar{G}}
\newcommand{\cX}{\mathcal{X}}
\newcommand{\cY}{\mathcal{Y}}
\newcommand{\cT}{\mathcal{T}}
\newcommand{\cZ}{\mathcal{Z}}
\newcommand{\cI}{\mathcal{I}}
\DeclareMathOperator{\supp}{supp}
\DeclareMathOperator{\IB}{IB}
\DeclareMathOperator{\PF}{PF}
\newcommand{\clique}{\textsc{Clique}}
\newcommand{\IBT}{\textsc{IB-Threshold}}
\newcommand{\IBL}{\textsc{IB-Lagrangian}}
\newcommand{\PFT}{\textsc{PF-Threshold}}
\newcommand{\PFL}{\textsc{PF-Lagrangian}}
\newcommand{\remove}[1]{}

\title{The Information Bottleneck and the Privacy Funnel are  NP-hard}
\author{ Roberto Bruno \  \
Ugo Vaccaro \\
Department of Computer Science, University of Salerno\\
84084 Fisciano (SA), Italy\\
Email:{\tt \{rbruno,uvaccaro\}@unisa.it}}
\date{}

\begin{document}
\maketitle

\begin{abstract}
Given a joint distribution $P_{XY}$ of random variables $X,Y$ on  finite alphabets, the information
bottleneck method seeks a representation $T$ of $X$ that is as informative
as possible about $Y$ while compressing $X$; 
the privacy funnel  asks for a representation that discloses much about
$X$ while leaking little about $Y$. In both scenarios  $T$ and $Y$  are
conditionally independent given $X$. We show that computing the optimal
value $\IB(P_{XY},R) = \sup_T \{ I(Y;T) : I(X;T) \le R \}$ of the
information bottleneck problem, as well as its Lagrangian relaxation
$L_\beta(P_{XY}) = \sup_T [ I(Y;T) - \beta I(X;T) ]$, is NP-hard.
Moreover, unless $\mathrm{P} = \mathrm{NP}$, neither $\IB(P_{XY},R)$ nor
$L_\beta(P_{XY})$ can be approximated in polynomial time within an
additive error of order $|\cX|^{-(1+\alpha)}$, for every fixed $\alpha > 0$.
We also prove that computing the
 privacy funnel
$\PF(P_{XY},r) = \inf_T \{ I(Y;T) : I(X;T) \ge r \}$ and its Lagrangian
relaxation $\Lambda_\beta(P_{XY}) = \inf_T [ I(Y;T) - \beta I(X;T) ]$ is
NP-hard, and that, unless $\mathrm{P} = \mathrm{NP}$, neither can be
approximated in polynomial time within an additive error of order
$|\cX|^{-2}$.
\end{abstract}


\section{Introduction}\label{sec:intro}

Extracting compact representations that preserve information relevant to
a task is a central problem in information theory, statistics, 
machine learning, and beyond. The information bottleneck (IB) method, introduced by
Tishby, Pereira and Bialek in 1999~\cite{TPB99}\footnote{We
should note that an equivalent optimization problem has previously been
studied in the literature on information
theory~\cite{Abu, WZ73a,WZ73b,WW75}.}, formulates this problem as a trade-off
between compressing an observed random variable $X$ and retaining
information about a target variable $Y$.

Applications of IB include clustering and document
classification~\cite{TPB99}, as well as the design of representations for
neural classifiers and generative models~\cite{GP20}. Its operational
interpretation through remote source coding under logarithmic loss,
together with extensions to distributed learning and communication,
connects these applications to information-theoretic performance
limits~\cite{ZEAS20}. The IB trade-off thus serves both as an objective
for constructing representations and as a benchmark for assessing their
quality. The vast literature on the information bottleneck problem and
its numerous applications cannot be synthesized here; we refer the reader
to the  review papers~\cite{AC20,GP20,HLYY24, Lew, Liu,  Pia, ZEAS20,ZL24} for a partial overview of the area.

Formally, for a joint distribution $P_{XY}$ on finite alphabets
$\cX \times \cY$ and a budget $R \ge 0$, the IB problem corresponds to the
following optimization problem:
\begin{equation}\label{eq:IB_def}
    \IB(P_{XY},R) = \sup_{\substack{P_{T\mid X}:\, T-X-Y\\ I(X;T) \le R}} I(Y;T).
\end{equation}
The supremum ranges over all finite representation alphabets $\cT$ and
stochastic encoders $P_{T\mid X}$; the Markov condition $T - X - Y$ means
that $T$ and $Y$ are conditionally independent given $X$.
We also consider the Lagrangian formulation
\begin{equation}\label{eq:L_def}
    L_\beta(P_{XY}) = \sup_{P_{T\mid X}:\, T-X-Y}
    \big[ I(Y;T) - \beta I(X;T) \big], \qquad \beta > 0.
\end{equation}
All information quantities are measured in bits.

In general, computing the optimal values in~\eqref{eq:IB_def}
and~\eqref{eq:L_def} is computationally challenging
\cite{Agmon2023, HuangElGamal2021}. 
Both problems are non-convex optimization problems
\cite{AC20,GP20, TPB99}. The iterative algorithm proposed
in~\cite{TPB99}, as well as its later variants, is only guaranteed to
converge to a local optimum~\cite{FriedmanMST01,StrouseS17,HuangElGamal2021}.
A
basic question is therefore whether the optimal values
in~\eqref{eq:IB_def} and~\eqref{eq:L_def} admit efficient computation or
approximation with provable accuracy for arbitrary finite input
distributions.

To the best of our knowledge, exact characterizations of the information
bottleneck curve are known in a handful of cases only. For binary sources
observed through a binary symmetric channel, the exact trade-off follows
from Mrs.\ Gerber's lemma~\cite{WZ73a}; extensions cover binary-input
memoryless symmetric channels with uniform input~\cite{CS89,SSZ05}. For
nonbinary Hamming channels with uniform input, Benger, Asoodeh and
Chen~\cite{BAC23} characterize the entire curve as the upper concave
envelope of an explicit scalar function. Other exact cases include
discrete deterministic targets $Y = f(X)$, for which
$\IB(P_{XY},R) = \min\{R, H(Y)\}$ for all $R \ge 0$~\cite[Sec.~3]{KTV19},
and the case in which $Y$ is obtained from $X$ through an erasure channel
with erasure probability $\epsilon$, independently of the input symbol.
In this case, $\IB(P_{XY},R) = (1-\epsilon)R$ for $0 \le R \le H(X)$, as
stated in~\cite[Lemma~1]{AC20}. In the continuous setting, Chechik \textit{et
al}.~\cite{CGTW05} obtain an analytic solution for jointly Gaussian
vectors, with optimal representations given by noisy linear projections
onto canonical-correlation directions. Dikshtein, Weinberger and
Shamai~\cite{DWS22} determine the optimal asymptotic relevance rate for
jointly stationary Gaussian processes through spectral water filling.

A closely related problem is the privacy funnel (PF), introduced by
Makhdoumi et al.~\cite{MSFM14}, in which one seeks a representation $T$
of $X$ that reveals at least a prescribed amount of information about $X$,
while leaking as little information as possible about $Y$. Formally, for
a rate $0\leq r \leq H(X)$, the PF problem corresponds to the following optimization problem:
\begin{equation}\label{eq:PF_def}
    \PF(P_{XY},r) = \inf_{\substack{P_{T\mid X}:\, T-X-Y\\ I(X;T) \ge r}} I(Y;T),
\end{equation}
and its Lagrangian formulation is
\begin{equation}\label{eq:Lambda_def}
    \Lambda_\beta(P_{XY}) = \inf_{P_{T\mid X}:\, T-X-Y}
    \big[ I(Y;T) - \beta I(X;T) \big], \qquad \beta > 0.
\end{equation}
Problems \eqref{eq:PF_def}  and \eqref{eq:Lambda_def} share the same computational difficulties as
\eqref{eq:IB_def}  and \eqref{eq:L_def}.

\subsection{Our results}
Our results provide a complexity-theoretic explanation for the
computational difficulties discussed above, by establishing a
fundamental obstacle to efficient computation in the worst case.
For the information bottleneck, we prove that the threshold decision
problems associated with both~\eqref{eq:IB_def} and~\eqref{eq:L_def}
are NP-hard. Thus, unless $\mathrm{P}=\mathrm{NP}$, no polynomial-time
algorithm can solve these problems for all instances.

 The NP-hardness of~\eqref{eq:L_def} is first obtained with the
Lagrange multiplier fixed at $\beta = 1/2$, and then extended to every
fixed rational $\beta \in (0,1)$, and to every fixed real
$\beta \in (0,1)$ computable in polynomial time. Moreover, unless
$\mathrm{P} = \mathrm{NP}$, no polynomial-time algorithm can approximate
$L_{1/2}(P_{XY})$ within additive error $1/(54|\cX|^{3/2})$, or
$\IB(P_{XY},R)$ within additive error $1/(144|\cX|^{3/2})$
(Theorems~\ref{thm:IBL_hard} and~\ref{thm:IBT_hard}). By suitably modifying the source instances 
of the decision problem from which we perform the reduction, 
we successively improve the exponent
$3/2$ as follows: for every fixed $\alpha > 0$, unless $\mathrm{P} = \mathrm{NP}$, no
polynomial-time algorithm can approximate $L_{1/2}(P_{XY})$ within
$1/(36|\cX|^{1+\alpha})$, or $\IB(P_{XY},R)$ within
$1/(96|\cX|^{1+\alpha})$, and the same holds, up to a constant factor, for
every fixed rational $\beta \in (0,1)$ (Theorem~\ref{thm:padding}).

Building on the same methodology, we show that computing the optimal
value of the privacy funnel is NP-hard as well. This hardness applies not
only to the exact computation, but also to approximation within an
additive error of order $|\cX|^{-2}$ (Theorem~\ref{thm:PF_hard}).
Moreover, the same results hold for the Lagrangian
formulation~\eqref{eq:Lambda_def}, for every fixed rational multiplier
$\beta \in (0,1)$ (Proposition~\ref{prop:PF_lagrangian}).

All our results hold in the standard bit model of computation, with
rational input distributions and rational rate budgets.

\subsection{Overview of the proofs}

All the results are proved by polynomial-time reductions from the
\clique{} problem, and the reductions share a common construction. Given an instance
$(G,k)$ of \clique{}, we construct a larger graph $\G$ by inserting, for
each edge $e$ of $G$, a clique $C_e$ of size $k-1$ containing only the
edge $e$ among the edges of $G$. This construction ensures that the
cliques $C_e$ identify a partition of the edges of $\G$, and that no clique of size
larger than $k-1$ is created, apart from those already present in $G$. In
all our reductions, the random variable $X$ is uniformly distributed on the oriented edges of
$\G$. We then consider two encoders of $X$: a deterministic encoder
$T_0$, which maps each oriented edge to the clique $C_e$ containing it,
and, if $G$ contains a $k$-clique $K$, a randomized encoder $T_1$ which
uses $K$.

For the information bottleneck, given $X = x$, the random variable $Y$ is
one of the two endpoints of the oriented edge $x$, chosen uniformly at
random. If $G$ has no $k$-clique, we show that $L_{1/2}(P_{XY})$ is
achieved by the encoder $T_0$. If, instead, $G$ has a $k$-clique, we show that 
the encoder $T_1$
strictly improves over $T_0$, and that the improvement is at least
inverse-polynomial in $|\cX|$. This yields both the NP-hardness of the
Lagrangian problem and the hardness of approximating $L_{1/2}(P_{XY})$.
The result for the constrained formulation \eqref{eq:IB_def} follows from the same
construction, by a suitable choice of the rate budget $R$. Moreover, by
modifying the instance $(G,k)$ of \clique{} without changing its answer,
we further improve the exponent $3/2$ in
Theorems~\ref{thm:IBL_hard} and~\ref{thm:IBT_hard} to $1 + \alpha$, for
every fixed $\alpha > 0$.

For the privacy funnel, we keep the same  $X$, $T_0$ and $T_1$ as for the IB problem,
but we change the
channel from $X$ to $Y$: given $X$, the random variable $Y$ either
reveals $X$, or reveals a pair of oriented edges that cannot belong to a
common clique. The reduction relies on the lower bound 
$I(Y;T) \ge I(X;T) - H(X \mid Y)$, which holds for every encoder $T$. More precisely, we show that, if $G$ does not contain a $k$-clique, then the lower bound cannot be attained at any rate smaller than
$\log|E(G)|$, while, if $G$ contains a $k$-clique, it is attained at a strictly
smaller rate by combining the encoders $T_0$ and $T_1$.

\subsection{Organization of the paper}

In Section~\ref{sec:prelim}, we set the notation, recall some basic facts on
encoders, and formally define the decision problems we study. In
Section~\ref{sec:construction}, we describe the extended graph $\G$ and
the encoders $T_0$ and $T_1$, which are common to all our reductions. In
Section~\ref{sec:IB}, we prove the hardness of the information
bottleneck, both in its Lagrangian and in its constrained formulation. In
Section~\ref{sec:PF}, we prove the hardness of the privacy funnel.
Section~\ref{sec:conclusions} contains some concluding remarks and open
problems.

\section{Preliminaries}\label{sec:prelim}

\subsection{Notation}\label{subsec:notation}

All alphabets throughout the paper are finite, and we denote by $\log$
the logarithm in base $2$ and by $\ln$ the natural logarithm. Let $X$ be
a random variable taking values in an alphabet $\cX$. We denote by $P_X$
the probability distribution of $X$, that is,
$P_X(x) = \mathbb{P}(X = x)$ for each $x \in \cX$, and by
\begin{equation}\label{eq:entropy}
    H(X) = -\sum_{x \in \cX} P_X(x) \log P_X(x)
\end{equation}
the Shannon entropy of $X$, with the convention $0 \log 0 = 0$. Since
$H(X)$ depends only on the distribution $P_X$, we also write $H(P_X)$ to
denote the entropy of the distribution $P_X$. We denote by
$h_2(z) = -z\log z - (1-z)\log(1-z)$, for $z \in [0,1]$, the binary
entropy function.

Let $(X,Y)$ be a pair of random variables taking values in
$\cX \times \cY$, with joint probability distribution
$P_{XY}(x,y) = \mathbb{P}(X = x, Y = y)$ and marginal distributions $P_X$
and $P_Y$. Without loss of generality, we assume that $P_X(x) > 0$ for
all $x \in \cX$; otherwise, we simply remove from $\cX$ the values with
$P_X(x) = 0$. Then, for every $x \in \cX$, the conditional distribution $P_{Y\mid X}(\cdot\mid x)$
of $Y$ given $X = x$ is
\begin{equation}\label{eq:cond_distr}
    P_{Y\mid X}(y \mid x) = \frac{P_{XY}(x,y)}{P_X(x)},
    \qquad \forall y \in \cY,
\end{equation}
and we denote by $H(Y \mid X = x) = H\big(P_{Y\mid X}(\cdot \mid x)\big)$
its entropy. The conditional entropy of $Y$ given $X$ is the average of
$H(Y \mid X = x)$ with respect to $P_X$, that is,
\begin{equation}\label{eq:conditional_entropy}
    H(Y \mid X) = \sum_{x \in \cX} P_X(x)\, H(Y \mid X = x),
\end{equation}
and the mutual information between $X$ and $Y$ is
\begin{equation}\label{eq:MI}
    I(X;Y) = H(Y) - H(Y \mid X) = H(X) - H(X \mid Y).
\end{equation} 
Conditional
entropies and mutual information given a third random variable $Z$ are
defined analogously;  we recall that
$I(X;Y \mid Z) = H(X \mid Z) - H(X \mid Y,Z)$, and the chain rules
\begin{equation}\label{eq:chain_rule}
    H(X,Y) = H(X) + H(Y \mid X)
    \qquad\text{and}\qquad
    I(X; Y,Z) = I(X;Z) + I(X; Y \mid Z).
\end{equation}
Let $(X,Y,T)$ be a triple of random variables with joint
distribution $P_{XYT}$ on $\cX \times \cY \times \cT$. The random
variables $Y$ and $T$ are conditionally independent given $X$ if
\begin{equation}\label{eq:cond_indep}
    P_{YT\mid X}(y,t \mid x) = P_{Y\mid X}(y \mid x)\, P_{T\mid X}(t \mid x),
    \qquad \forall (x,y,t) \in \cX \times \cY \times \cT.
\end{equation}

Let $P$ and $Q$ be two probability distributions on a finite set $S$,
and let $\supp(P) = \{s \in S : P(s) > 0\}$ denote the support of $P$. If
$\supp(P) \subseteq \supp(Q)$, the Kullback--Leibler divergence and the
$\chi^2$-divergence of $P$ from $Q$ are defined as
\begin{equation}\label{eq:divergences}
    D(P \,\|\, Q) = \sum_{s \in \supp(P)} P(s) \log\frac{P(s)}{Q(s)}
    \qquad\text{and}\qquad
    \chi^2(P \,\|\, Q) = \sum_{s \in \supp(Q)} \frac{\big(P(s) - Q(s)\big)^2}{Q(s)}.
\end{equation}
Both divergences are non-negative, and they vanish if and only if
$P = Q$. Moreover, the two divergences are related by the following well-known
inequality.

\begin{proposition}[\cite{SV16}]\label{prop:KL_chi2}
    Let $P$ and $Q$ be two probability distributions on a finite set $S$
    such that $\supp(P) \subseteq \supp(Q)$. Then, it holds that
    \begin{equation}\label{eq:KL_chi2}
        D(P \,\|\, Q) \le \log\big(1 + \chi^2(P \,\|\, Q)\big)
        \le \frac{\chi^2(P \,\|\, Q)}{\ln 2}.
    \end{equation}
\end{proposition}

Finally, given a probability distribution $P$ on an alphabet $S$, we denote by
$P \otimes P$ the product distribution on $S \times S$, defined by
$(P \otimes P)(s,s') = P(s)P(s')$ for all $(s,s') \in S \times S$.

\subsection{Encoders}\label{subsec:encoders}

Let $(X,Y) \sim P_{XY}$ be a pair of random variables taking values in
$\cX \times \cY$. In the optimization problems~\eqref{eq:IB_def}--\eqref{eq:Lambda_def},
one seeks a representation (or encoder) $T$ of $X$ that is conditionally
independent of $Y$ given $X$. This type of encoder is formalized as
follows.

\begin{definition}\label{def:encoder}
    Let $\cT$ be a finite set. An encoder of $X$ is a conditional
    distribution (or channel) $P_{T\mid X}$ from $\cX$ to $\cT$, that is, a
    collection $\{P_{T\mid X}(\cdot \mid x)\}_{x \in \cX}$ of probability
    distributions on $\cT$. The encoder $P_{T\mid X}$ induces a random
    variable $T$ taking values in $\cT$ such that $(X,Y,T)$ has joint
    distribution
    \begin{equation}\label{eq:encoder_joint}
        P_{XYT}(x,y,t) = P_X(x)\, P_{Y\mid X}(y \mid x)\, P_{T\mid X}(t \mid x),
        \qquad \forall (x,y,t) \in \cX \times \cY \times \cT.
    \end{equation}
    In particular, $T$ and $Y$ are conditionally independent given $X$. Moreover, we refer to the mutual information $I(X;T)$ as the
    \emph{rate} of the encoder $T$.
\end{definition}

To simplify the notation, we identify an encoder with the random variable
$T$ it induces, and we refer to $T$ itself as the encoder whenever the
channel $P_{T\mid X}$ is clear from the context. Moreover, without loss of
generality, we assume that $P_T(t) > 0$ for all $t \in \cT$, and we call
the distributions $\{P_{X\mid T}(\cdot \mid t)\}_{t \in \cT}$ the
\emph{posteriors} of $T$.

Let $W$ denote the channel from $\cX$ to $\cY$ induced by $P_{XY}$, that
is, $W(y \mid x) = P_{Y\mid X}(y \mid x)$ for all $x \in \cX$ and
$y \in \cY$. The channel $W$ maps every probability distribution $Q$ on
$\cX$ to the probability distribution $WQ$ on $\cY$ defined by
\begin{equation}\label{eq:WQ}
    (WQ)(y) = \sum_{x \in \cX} W(y \mid x)\, Q(x),
    \qquad \forall y \in \cY.
\end{equation}
In particular, $WP_X = P_Y$. The next lemma shows that the same relation holds
conditionally on $T$. Indeed,  since $T$ and $Y$ are conditionally independent
given $X$, applying $W$ to the posterior of $X$ given $T = t$ gives the
posterior of $Y$ given $T = t$.

\begin{lemma}\label{lemma:posterior_Y_given_T}
    Let $T$ be an encoder of $X$. Then, it holds that
    \begin{equation}\label{eq:posterior_Y_given_T}
        P_{Y\mid T}(\cdot \mid t) = W P_{X\mid T}(\cdot \mid t),
        \qquad \forall t \in \cT.
    \end{equation}
\end{lemma}

\begin{proof}
    Let $y \in \cY$ and $t \in \cT$. By~\eqref{eq:encoder_joint}, and
    since $P_X(x)\, P_{T\mid X}(t \mid x) = P_T(t)\, P_{X\mid T}(x \mid t)$,
    we have that
    \begin{equation*}
        P_{Y\mid T}(y \mid t)
        = \frac{\sum_{x \in \cX} P_X(x)\, P_{T\mid X}(t \mid x)\, W(y \mid x)}{P_T(t)}
        = \sum_{x \in \cX} W(y \mid x)\, P_{X\mid T}(x \mid t)
        = \big(W P_{X\mid T}(\cdot \mid t)\big)(y),
    \end{equation*}
    where the last equality follows from~\eqref{eq:WQ}.
\end{proof}

The next simple lemma gives a lower bound on the information that an encoder
leaks about $Y$, in terms of the information it reveals about $X$. It
plays a useful role in Section~\ref{sec:PF}.

\begin{lemma}\label{lemma:PF_lower}
    For every joint distribution $P_{XY}$ and every encoder $T$ of $X$,
    \begin{equation}\label{eq:PF_lower}
        I(Y;T) \ge I(X;T) - H(X \mid Y),
    \end{equation}
    with equality if and only if $H(X \mid Y,T) = 0$.
\end{lemma}

\begin{proof}
    Since $Y$ and $T$ are conditionally independent given $X$, we have
    that $I(X;T) = I(X,Y;T) = I(Y;T) + I(X;T \mid Y)$. Moreover,
    $I(X;T \mid Y) = H(X \mid Y) - H(X \mid Y,T) \le H(X \mid Y)$, with
    equality if and only if $H(X \mid Y,T) = 0$. Putting these together,
    we get that
    \begin{equation*}
        I(Y;T) = I(X;T) - I(X;T \mid Y) \ge I(X;T) - H(X \mid Y),
    \end{equation*}
    which proves~\eqref{eq:PF_lower}, with equality if and only if
    $H(X \mid Y,T) = 0$.
\end{proof}

Finally, we show that two encoders can be combined into a new encoder,
whose rate interpolates between their rates and whose posteriors are
those of the two encoders. We will use this fact in
Section~\ref{sec:PF}.

\begin{lemma}\label{lemma:mixture}
    Let $T_a$ and $T_b$ be two encoders of $X$ taking values in $\cT_a$
    and $\cT_b$, respectively, and let $\theta \in [0,1]$. Let $T_\theta$
    be the encoder of $X$ taking values in
    $(\{a\} \times \cT_a) \cup (\{b\} \times \cT_b)$ defined by
    \begin{equation}\label{eq:mixture}
        P_{T_\theta \mid X}\big((a,t) \mid x\big) = \theta\, P_{T_a \mid X}(t \mid x),
        \qquad
        P_{T_\theta \mid X}\big((b,t) \mid x\big) = (1-\theta)\, P_{T_b \mid X}(t \mid x).
    \end{equation}
    In other words, $T_\theta$ applies $T_a$ with probability $\theta$ and
    $T_b$ with probability $1 - \theta$, independently of $X$, and reveals
    which of the two encoders has been applied. Then, the posteriors of
    $T_\theta$ are posteriors of $T_a$ or of $T_b$, that is,
    \begin{equation}\label{eq:mixture_post}
        P_{X\mid T_\theta}\big(\cdot \mid (a,t)\big) = P_{X\mid T_a}(\cdot \mid t)
        \quad\text{and}\quad
        P_{X\mid T_\theta}\big(\cdot \mid (b,t)\big) = P_{X\mid T_b}(\cdot \mid t),
    \end{equation}
    whenever the left-hand sides are defined. Moreover,
    \begin{equation}\label{eq:mixture_rate}
        I(X;T_\theta) = \theta\, I(X;T_a) + (1-\theta)\, I(X;T_b).
    \end{equation}
\end{lemma}

\begin{proof}
    Let $J$ be a random variable taking values in $\{a,b\}$, independent
    of $(X,Y)$, with $\mathbb{P}(J = a) = \theta$. Given $X$ and $J = j$,
    we generate $T_j$ through the channel $P_{T_j \mid X}$, and we let
    $T_\theta = (J, T_J)$. Then $T_\theta$ has the channel defined in~\eqref{eq:mixture},
    and it is an encoder of $X$, since it depends only on $X$ and on $J$,
    which is independent of $(X,Y)$.

    We first prove~\eqref{eq:mixture_post}. Let $j \in \{a,b\}$ with
    $\mathbb{P}(J = j) > 0$, and let $t$ be such that $P_{T_j}(t) > 0$.
    By~\eqref{eq:mixture}, we have that
    $P_{T_\theta}\big((j,t)\big) = \mathbb{P}(J = j)\, P_{T_j}(t)$, and
    hence, by Bayes' rule, for every $x \in \cX$, it holds that
    \begin{equation*}
        P_{X\mid T_\theta}\big(x \mid (j,t)\big)
        = \frac{P_X(x)\, \mathbb{P}(J = j)\, P_{T_j \mid X}(t \mid x)}
               {\mathbb{P}(J = j)\, P_{T_j}(t)}
        = P_{X\mid T_j}(x \mid t),
    \end{equation*}
    which proves~\eqref{eq:mixture_post}.
    
    We now prove~\eqref{eq:mixture_rate}. Since $J$ is independent of $X$,
    we have that $I(X;J) = 0$. Moreover, given $J = j$, the pair
    $(X, T_J)$ has the joint distribution induced by the encoder $T_j$, so
    that $I(X;T_J \mid J = j) = I(X;T_j)$. Hence, by the chain
    rule~\eqref{eq:chain_rule}, we obtain that
    \begin{equation*}
        I(X;T_\theta) = I(X;J) + I(X;T_J \mid J)
        = \theta\, I(X;T_a) + (1-\theta)\, I(X;T_b). \qedhere
    \end{equation*}
\end{proof}

\subsection{Decision problems}\label{subsec:decision}

In this section, we state the decision versions of the optimization
problems~\eqref{eq:IB_def}--\eqref{eq:Lambda_def}, which we use to prove
NP-hardness\footnote{We refer to \cite{GJ79} for the basics about NP-hardness.}.
Throughout, a probability distribution is called rational if all its
probabilities are rational numbers.

\medskip
\noindent
{\bf\IBT}

\noindent
\textsc{Input:} A rational probability distribution $P_{XY}$ on a finite
alphabet $\cX \times \cY$, and non-negative rational numbers $R$ and
$\vartheta$.

\noindent
\textsc{Question:} Does there exist an encoder $T$ of $X$ such that
$I(X;T) \le R$ and $I(Y;T) \ge \vartheta$?

\medskip
\noindent
{\bf\IBL$_\beta$}, for a fixed $\beta > 0$

\noindent
\textsc{Input:} A rational probability distribution $P_{XY}$ on a finite
alphabet $\cX \times \cY$, and a rational number $\tau$.

\noindent
\textsc{Question:} Does there exist an encoder $T$ of $X$ such that
$I(Y;T) - \beta I(X;T) \ge \tau$?

\medskip
\noindent
{\bf\PFT}

\noindent
\textsc{Input:} A rational probability distribution $P_{XY}$ on a finite
alphabet $\cX \times \cY$, and non-negative rational numbers $r$ and
$\vartheta$.

\noindent
\textsc{Question:} Does there exist an encoder $T$ of $X$ such that
$I(X;T) \ge r$ and $I(Y;T) \le \vartheta$?

\medskip
\noindent
{\bf\PFL$_\beta$}, for a fixed $\beta > 0$

\noindent
\textsc{Input:} A rational probability distribution $P_{XY}$ on a finite
alphabet $\cX \times \cY$, and a rational number $\tau$.

\noindent
\textsc{Question:} Does there exist an encoder $T$ of $X$ such that
$I(Y;T) - \beta I(X;T) \le \tau$?
\medskip

\noindent
We write \IBL{} for \IBL$_{1/2}$. We observe that the optima in~\eqref{eq:IB_def},
\eqref{eq:L_def} and~\eqref{eq:Lambda_def}, as well as the optimum
in~\eqref{eq:PF_def} whenever $r \le H(X)$, are attained (see,
e.g.,~\cite{AC20,WW75}). Moreover, if $r > H(X)$, then no encoder satisfies $I(X;T) \ge r$, and,
in this case, we adopt the convention $\PF(P_{XY},r) = +\infty$. Hence, the above decision problems are equivalent to
comparing the optimal values of the corresponding optimization problems
with a threshold. Namely, \IBT{} asks whether
$\IB(P_{XY},R) \ge \vartheta$, \IBL$_\beta$ asks whether
$L_\beta(P_{XY}) \ge \tau$, and, similarly, \PFT{} and \PFL$_\beta$ ask
whether $\PF(P_{XY},r) \le \vartheta$ and $\Lambda_\beta(P_{XY}) \le \tau$,
respectively.

We prove the NP-hardness of these problems by reductions from the
well-known \clique{} decision problem~\cite{Karp72}, which we now recall. A graph
$G$ is a pair $(V(G),E(G))$, where $V(G)$ is a finite set of vertices and
$E(G)$ is a set of unordered pairs $\{u,v\}$ of distinct vertices
$u,v \in V(G)$, called edges. A clique of $G$ is a subset
$C \subseteq V(G)$ such that $\{u,v\} \in E(G)$ for all distinct
$u,v \in C$. The size of a clique $C$ is its cardinality $|C|$, and a
clique of size $k$ is called a $k$-clique. The clique number of $G$ is
\begin{equation}\label{eq:clique_number_def}
    \omega(G) = \max\{ |C| : C \subseteq V(G) \text{ is a clique of } G \}.
\end{equation}

\medskip
\noindent
\clique

\noindent
\textsc{Input:} A graph $G$ and a positive integer $k$.

\noindent
\textsc{Question:} Does $G$ contain a $k$-clique, that is, is
$\omega(G) \ge k$?
\medskip

\noindent
It is well known that \clique{} is NP-complete~\cite{Karp72}.

\section{The extended graph and the encoders \texorpdfstring{$T_0$}{T0} and \texorpdfstring{$T_1$}{T1}}\label{sec:construction}

In this section, we describe the construction shared by all our
reductions: the extended graph $\G$ of an instance $(G,k)$ of \clique{},
the random variable $X$, which is uniformly distributed on the oriented
edges of $\G$, and the two encoders $T_0$ and $T_1$ of $X$. The random
variable $Y$, and hence the joint distribution $P_{XY}$, will be defined
differently in the reductions for the information bottleneck (Section~\ref{sec:IB}) and for
the privacy funnel (Section~\ref{sec:PF}). However, since the encoders $T_0$ and $T_1$
are channels from $X$, the results of this section do not depend on the random variable $Y$.

\subsection{Oriented edges and cliques}

Let $G$ be a graph with at least one edge, and let
\begin{equation}\label{eq:oriented_edges}
    A(G) = \big\{ (u,v) \in V(G) \times V(G) : \{u,v\} \in E(G) \big\}
\end{equation}
denote the set of the $2|E(G)|$ oriented edges of $G$. For a clique $C$
of $G$ of size $j \ge 2$, let
\begin{equation}\label{eq:oriented_edges_clique}
    A(C) = \{ (u,v) \in C \times C : u \ne v \} \subseteq A(G)
\end{equation}
be the set of its $j(j-1)$ oriented edges, and let $U_C$ denote the
probability distribution on $A(G)$ that is uniform on $A(C)$, that is,
\begin{equation}\label{eq:U_C}
    U_C(x) =
    \begin{cases}
        \dfrac{1}{j(j-1)}, & x \in A(C),\\[6pt]
        0, & \text{otherwise}.
    \end{cases}
\end{equation}
Clearly, $H(U_C) = \log\big(j(j-1)\big)$.

\subsection{The extended graph}\label{subsec:extended_graph}

Let $(G,k)$ be an instance of \clique{}, and let
\begin{equation}\label{eq:n_N_B}
    n = |V(G)|, \qquad N = |E(G)|, \qquad B = \binom{k}{2} = \frac{k(k-1)}{2}.
\end{equation}
We assume that $G$ is represented by its adjacency matrix, so that the
size of $(G,k)$ is at least $n$. For the sake of simplicity, we restrict
to instances of \clique{} satisfying
\begin{equation}\label{eq:assumption_on_G}
    k \ge 4 \qquad\text{and}\qquad N \ge B.
\end{equation}
This does not affect the NP-hardness. Indeed, if $k \le 3$, \clique{} can
be solved in $O(n^3)$ time by examining all subsets of at most three
vertices. On the other hand, if $N < B$, then $G$ has no $k$-clique, since
a $k$-clique has $B$ edges.

Let $m = k - 1 \ge 3$. We construct the extended graph $\G$ from $G$ as
follows: for each edge $e = \{u,v\} \in E(G)$, we add a set $S_e$ of
$m - 2$ new vertices, and we add all the edges needed to make
$C_e = \{u,v\} \cup S_e$ a clique of size $m$. No other vertices or edges
are added. Throughout the paper, we let
\begin{equation}\label{eq:s_mu_M}
    s = m(m-1), \qquad \mu = \frac1s, \qquad M = |A(\G)|.
\end{equation}

The next lemma collects the properties of $\G$ on which our reductions
rely.
\begin{lemma}\label{lemma:gadget}
    The graph $\G$ satisfies the following properties.
    \begin{enumerate}
        \item For every $e \in E(G)$, the only edge of $G$ in $E(C_e)$ is
        $e$. Moreover, the collection $\{E(C_e)\}_{e \in E(G)}$ is a
        partition of $E(\G)$. Consequently, every oriented edge of $\G$
        belongs to $A(C_e)$ for exactly one $e \in E(G)$.
        \item The number of oriented edges of $\G$ is 
        \begin{equation}\label{eq:M_A_G_def}
            M = |A(\G)| = Nm(m-1) = Ns.
        \end{equation}
        \item $\omega(\G) = \max\{\omega(G), m\}$.
        \item $|V(\G)| \le n + n^3$, $M \le n^4$, and $\G$ can be
        constructed in time polynomial in $n$.
    \end{enumerate}
\end{lemma}

\begin{proof}
    1. Let $e \in E(G)$. Since $C_e \cap V(G) = e$, the only edge of $G$
    with both endpoints in $C_e$ is $e$, which proves the first claim. We
    now show that $\{E(C_e)\}_{e \in E(G)}$ is a partition of $E(\G)$. By
    construction, every edge of $\G$ belongs to $E(C_e)$ for some
    $e \in E(G)$. Thus, it suffices to show that this $e$ is unique. Let
    $x = \{u,v\} \in E(\G)$. If $x \in E(G)$, then, by the first claim,
    $x \in E(C_e)$ if and only if $x = e$. If $x \notin E(G)$, then $x$ was
    added in the construction of $\G$, and at least one of its endpoints,
    say $u$, belongs to $S_e$ for some $e \in E(G)$, so that
    $x \in E(C_e)$. Since the sets $S_e$ are pairwise disjoint,
    $u \notin C_{e'}$ for every $e' \ne e$, and hence $x \notin E(C_{e'})$ for every $e\neq e'$.
    In both cases, $x$ belongs to exactly one set $E(C_e)$. Therefore, each edge of $\G$ lies in exactly one of the sets
    $E(C_e)$, with $e \in E(G)$. Hence,  $\{E(C_e)\}_{e\in E(G)}$ forms a partition of $E(\G)$. Finally, since
    $(u,v) \in A(C_e)$ if and only if $\{u,v\} \in E(C_e)$, every oriented
    edge of $\G$ belongs to $A(C_e)$ for exactly one $e \in E(G)$.

    2. By property~1, $E(\G)$ is the disjoint union of the $N$ sets
    $E(C_e)$, each of cardinality $\binom{m}{2}$. Hence, the total number of oriented edges is
    $|A(\G)| = 2N\binom{m}{2} = Nm(m-1)=Ns$.

    3. We first show that $\omega(\G) \le \max\{\omega(G), m\}$. Let $K$
    be a clique of $\G$. If $K$ contains a vertex $z \in S_e$ for some
    $e \in E(G)$, then, by construction of $\G$, all the neighbours of $z$ belong
    to $C_e$. Therefore, $K \subseteq C_e$ and $|K| \le m$. Otherwise,  if $K$ does not contain a vertex $z \in S_e$ for any $e\in E(G)$, then we have that
    $K \subseteq V(G)$. Therefore, $K$ is a clique of $G$, and $|K| \le \omega(G)$.
    Conversely, every clique of $G$ is a clique of $\G$, and each $C_e$ is
    a clique of $\G$ of size $m$. Thus, $\omega(\G) \ge \max\{\omega(G), m\}$. Putting everything together we have that    $\omega(\G) = \max\{\omega(G), m\}$.

    4. We have that $N \le \binom{n}{2} \le n^2$, and,
    by~\eqref{eq:assumption_on_G}, $\binom{k}{2} = B \le N \le \binom{n}{2}$,
    which implies $k \le n$ and hence $m < n$. Since $\G$ is obtained from
    $G$ by adding $m-2$ new vertices for each edge of $G$, together with
    the edges of the cliques $C_e$, we get that
    $|V(\G)| = n + N(m-2) \le n + n^3$, and, by property~2,
    $M = Nm(m-1) \le n^4$. Therefore, $\G$ can be constructed in time
    polynomial in $n$.
\end{proof}

\subsection{The random variable \texorpdfstring{$X$}{X} and the encoder \texorpdfstring{$T_0$}{T0}}

In all our reductions, the random variable $X$ is uniformly distributed
on $\cX = A(\G)$, that is,
\begin{equation}\label{eq:P_X}
    P_X(x) = \frac1M = \frac{\mu}{N}, \qquad \forall x \in A(\G).
\end{equation}
By property~1 of Lemma~\ref{lemma:gadget}, we can define the
deterministic encoder $T_0$ of $X$, taking values in $E(G)$, by
\begin{equation}\label{eq:T_0_def}
    P_{T_0\mid X}(e \mid x) =
    \begin{cases}
        1, & \text{if } x \in A(C_e),\\
        0, & \text{otherwise},
    \end{cases}
    \qquad \forall x \in A(\G),\ e \in E(G).
\end{equation}
In other words, $T_0$ maps each oriented edge $x$ to the unique edge
$e \in E(G)$ such that $x \in A(C_e)$, that is, it reveals the clique
$C_e$ containing $x$. The following lemma computes the posteriors and the rate of $T_0$.

\begin{lemma}\label{lemma:T_0}
    The encoder $T_0$ defined in~\eqref{eq:T_0_def} satisfies
    \begin{equation}\label{eq:T_0_posteriors}
        P_{T_0}(e) = \frac1N
        \qquad\text{and}\qquad
        P_{X\mid T_0}(\cdot \mid e) = U_{C_e},
        \qquad \forall e \in E(G).
    \end{equation}
    Consequently, $H(X \mid T_0) = \log s$ and
    \begin{equation}\label{eq:T_0_rate}
        I(X;T_0) = \log N.
    \end{equation}
\end{lemma}

\begin{proof}
    We first prove~\eqref{eq:T_0_posteriors}. For every $e \in E(G)$, the
    set $A(C_e)$ contains $s = m(m-1)$ oriented edges, each of probability
    $1/M$ under $P_X$. Hence, by~\eqref{eq:T_0_def} and property~2 of
    Lemma~\ref{lemma:gadget}, we have that
    \begin{equation*}
        P_{T_0}(e) = \sum_{x \in A(\G)} P_X(x)\, P_{T_0\mid X}(e \mid x)
        = \frac{|A(C_e)|}{M} = \frac{s}{Ns} = \frac1N.
    \end{equation*}
    Moreover, by Bayes' rule, $P_{X\mid T_0}(x \mid e) = N/M = 1/s$ if
    $x \in A(C_e)$, and $P_{X\mid T_0}(x \mid e) = 0$ otherwise, that is,
    $P_{X\mid T_0}(\cdot \mid e) = U_{C_e}$, where $U_{C_e}$ is the probability distribution on $A(\G)$ that is uniform on $A(C_e)$  (cf.~\eqref{eq:U_C}).

    We now compute the rate of $T_0$. Since each $U_{C_e}$ is uniform on
    $s$ oriented edges, by~\eqref{eq:conditional_entropy} we get that
    $H(X \mid T_0) = \sum_{e \in E(G)} \frac1N \log s = \log s$. Moreover, since $X$
    is uniform on the $M$ elements of $A(\G)$, it follows that
    $I(X;T_0) = H(X) - H(X \mid T_0) = \log M - \log s = \log N$, where the last equality holds by property~2 of Lemma~\ref{lemma:gadget}.
\end{proof}

\subsection{The encoder \texorpdfstring{$T_1$}{T1}}\label{subsec:T_1}

We now assume that $G$ contains a $k$-clique $K$, and we define an
encoder $T_1$ that uses $K$. The encoder depends on a parameter
$p \in (0,1)$, which will be chosen differently for the information
bottleneck and for the privacy funnel. Let $T_1$ be the encoder taking
values in $\{\star\} \cup E(G)$, where $\star$ is a new symbol, defined by
\begin{equation}\label{eq:T_1_def}
    P_{T_1\mid X}(\star \mid x) =
    \begin{cases}
        p, & \text{if } x \in A(K),\\
        0, & \text{otherwise},
    \end{cases}
    \qquad
    P_{T_1\mid X}(e \mid x) = \big(1 - P_{T_1\mid X}(\star \mid x)\big)\, P_{T_0\mid X}(e \mid x),
\end{equation}
for every $x \in A(\G)$ and $e \in E(G)$. In other words, if
$X \notin A(K)$, then $T_1 = T_0$, while, if $X \in A(K)$, then
$T_1 = \star$ with probability $p$ and $T_1 = T_0$ with probability
$1-p$. Thus, $T_1$ differs from $T_0$ only on the oriented edges of the
$k$-clique $K$. Let
\begin{equation}\label{eq:mu_eps_lambda}
    \varepsilon = 2\mu p \qquad\text{and}\qquad \lambda = \frac{B\varepsilon}{N}.
\end{equation}
Note that $\varepsilon < 1$, since $\mu \le 1/6$ and $p < 1$. For every
$e = \{u,v\} \in E(K)$, let $A(e) = \{(u,v),(v,u)\}$ denote the set of the
two orientations of $e$, and let $V_e$ be the probability distribution on
$A(\G)$ defined by
\begin{equation}\label{eq:V_e_def}
    V_e(x) =
    \begin{cases}
        a_0 = \dfrac{\mu}{1-\varepsilon}, & \text{if } x \in A(C_e) \setminus A(e),\\[8pt]
        a_1 = \dfrac{\mu(1-p)}{1-\varepsilon}, & \text{if } x \in A(e),\\[8pt]
        0, & \text{otherwise}.
    \end{cases}
\end{equation}
The distribution $V_e$ is a slight perturbation of the uniform distribution $U_{C_e}$ \eqref{eq:T_0_posteriors}. Indeed, the
probability of the two orientations of $e$ is slightly decreased, and the
probability of the other oriented edges of $C_e$ is slightly increased. The next lemma computes the output distribution and the posteriors of
$T_1$, which we need to compare $T_1$ with $T_0$.

\begin{lemma}\label{lemma:T_1}
    The encoder $T_1$ defined in~\eqref{eq:T_1_def} satisfies
    \begin{equation}\label{eq:T_1_distribution}
        P_{T_1}(t) =
        \begin{cases}
            \lambda, & t = \star,\\
            (1-\varepsilon)/N, & t \in E(K),\\
            1/N, & t \in E(G) \setminus E(K),
        \end{cases}
        \qquad
        P_{X\mid T_1}(\cdot \mid t) =
        \begin{cases}
            U_K, & t = \star,\\
            V_t, & t \in E(K),\\
            U_{C_t}, & t \in E(G) \setminus E(K).
        \end{cases}
    \end{equation}
\end{lemma}

\begin{proof}
    Recall that $P_X(x) = 1/M = \mu/N$ for every $x \in A(\G)$,
    by~\eqref{eq:P_X}. We consider the cases $T_1=\star$ and $T_1=e$ with $e\in E(G)$,  
    separately.

    We first consider the case $T_1=\star$. Since $K$ is a clique of size
    $k$, it has $k(k-1) = 2B$ oriented edges. Hence, by the definition~\eqref{eq:T_1_def} of the encoder $T_1$, and
    since $\varepsilon = 2\mu p$ and $\lambda = \frac{B\varepsilon}{N}$
    by~\eqref{eq:mu_eps_lambda}, we have that
    \begin{equation*}
        P_{T_1}(\star)
        = \sum_{x \in A(\G)} P_X(x)\, P_{T_1 \mid X}(\star \mid x)
        = \sum_{x \in A(K)} \frac{\mu}{N}\, p
        = \frac{2B\mu p}{N}
        = \frac{B\varepsilon}{N} = \lambda,
    \end{equation*}
    and, by Bayes' rule,
    $$P_{X\mid T_1}(x \mid \star) =\frac{P_X(x)\, P_{T_1 \mid X}(\star \mid x)}{P_{T_1}(\star)}= \frac{\mu p/N}{2B\mu p/N} = \frac{1}{k(k-1)}$$
    for $x \in A(K)$, and $P_{X\mid T_1}(x \mid \star) = 0$ otherwise. By
    definition~\eqref{eq:U_C} of $U_K$, this means that
    $P_{X\mid T_1}(\cdot \mid \star) = U_K$.

    We now consider the case $T_1=e$ for $e \in E(G)$. By~\eqref{eq:T_0_def}
    and~\eqref{eq:T_1_def}, we have that
    \begin{equation}\label{eq:T_1_channel_e}
        P_{T_1\mid X}(e \mid x) =
        \begin{cases}
            1 - p, & \text{if } x \in A(C_e) \cap A(K),\\
            1, & \text{if } x \in A(C_e) \setminus A(K),\\
            0, & \text{if } x \notin A(C_e).
        \end{cases}
    \end{equation}
    Thus, $T_1$ differs from $T_0$ only on the oriented edges in
    $A(C_e) \cap A(K)$. We claim that
    \begin{equation}\label{eq:A_Ce_cap_A_K}
        A(C_e) \cap A(K) =
        \begin{cases}
            A(e), & \text{if } e \in E(K),\\
            \emptyset, & \text{otherwise}.
        \end{cases}
    \end{equation}
    Indeed, let $x = (u,v) \in A(K)$. Since $K$ is a clique of $G$,
    $\{u,v\}$ is an edge of $G$, and, by property~1 of
    Lemma~\ref{lemma:gadget}, it holds that  $\{u,v\} \in E(C_e)$ if and only if
    $\{u,v\} = e$. Hence, $A(C_e) \cap A(K)$ consists of the orientations
    of $e$ that belong to $A(K)$, which are both orientations if
    $e \in E(K)$, and none otherwise, proving~\eqref{eq:A_Ce_cap_A_K}.

    Now, if $e \in E(G) \setminus E(K)$, then, by~\eqref{eq:A_Ce_cap_A_K}
    and~\eqref{eq:T_1_channel_e},
    $P_{T_1\mid X}(e \mid \cdot) = P_{T_0\mid X}(e \mid \cdot)$, and, by
    Lemma~\ref{lemma:T_0}, $P_{T_1}(e) = 1/N$ and
    $P_{X\mid T_1}(\cdot \mid e) = U_{C_e}$. If $e \in E(K)$, then, by
    \eqref{eq:A_Ce_cap_A_K}, it holds that $A(C_e) \cap A(K) = A(e)$. Notice that  $A(e)$  contains $2$
    elements, while $A(C_e) \setminus A(e)$ contains $s - 2$ elements. Hence,
    by~\eqref{eq:T_1_channel_e}, and using $\mu s = 1$, we obtain that
    \begin{equation*}
        P_{T_1}(e) =\sum_{x \in A(\G)} P_X(x)\, P_{T_1 \mid X}(e \mid x)= \frac{\mu}{N}\big( (s-2) + 2(1-p) \big)
        = \frac{1 - 2\mu p}{N} = \frac{1-\varepsilon}{N}.
    \end{equation*}
    Moreover, applying  Bayes' rule, one can see that the following equalities hold:
    $P_{X\mid T_1}(x \mid e) = \frac{\mu/N}{(1-\varepsilon)/N} = a_0$ if
    $x \in A(C_e) \setminus A(e)$,
    $P_{X\mid T_1}(x \mid e) = \frac{\mu(1-p)/N}{(1-\varepsilon)/N} = a_1$ if
    $x \in A(e)$, and $P_{X\mid T_1}(x \mid e) = 0$ otherwise;  that is,
    $P_{X\mid T_1}(\cdot \mid e) = V_e$ by~\eqref{eq:V_e_def}.
\end{proof}
Lemma~\ref{lemma:T_1} shows that the encoder $T_1$ replaces part of the posteriors
$U_{C_e}$ of $T_0$ with the  posterior $U_K$. In doing so, however, it perturbs the posterior $U_{C_e}$ into
$V_e$ (cf. \eqref{eq:V_e_def}) for every $e \in E(K)$.
The next lemma shows that the perturbed posteriors $V_e$ are close to
$U_{C_e}$, namely, their divergence from $U_{C_e}$ is quadratic in
$\varepsilon$.

\begin{lemma}\label{lemma:perturbation}
    For every $e \in E(K)$, we have that
    $H(V_e) = \log s - D(V_e \,\|\, U_{C_e})$. Moreover,
    the Kullback--Leibler divergence $D(V_e \,\|\, U_{C_e})$ does not depend on $e$, and satisfies
    \begin{equation}\label{eq:D_V_bound}
        D(V_e \,\|\, U_{C_e}) \le \frac{\chi^2(V_e \,\|\, U_{C_e})}{\ln 2}
        \le \frac{\varepsilon^2}{2\mu(1-\varepsilon)^2 \ln 2}.
    \end{equation}
\end{lemma}

\begin{proof}
    We first prove the identity for $H(V_e)$. Since
    $\supp(V_e) \subseteq A(C_e)$ and $U_{C_e}(x) = \mu = 1/s$ for every
    $x \in A(C_e)$, we have that
    \begin{equation*}
        D(V_e \,\|\, U_{C_e}) = \sum_{x \in \supp(V_e)} V_e(x) \log\frac{V_e(x)}{\mu}
        = -H(V_e) + \log\frac1\mu = \log s - H(V_e).
    \end{equation*}
    By~\eqref{eq:V_e_def}, $V_e$ takes the value $a_0$ on $s-2$ oriented
    edges and the value $a_1$ on $2$ oriented edges, and these values do
    not depend on $e$. Hence, $H(V_e)$, and thus $D(V_e \,\|\, U_{C_e})$,
    does not depend on $e$.

    We now prove~\eqref{eq:D_V_bound}. The first inequality  in~~\eqref{eq:D_V_bound} follows from
    Proposition~\ref{prop:KL_chi2}, since
    $\supp(V_e) \subseteq A(C_e) = \supp(U_{C_e})$. For the second one, we
    compute $\chi^2(V_e \,\|\, U_{C_e})$. Using $\mu p = \varepsilon/2$, we
    have that
    \begin{equation}\label{eq:a_0_a_1_}
        a_0 - \mu = \frac{\mu\varepsilon}{1-\varepsilon}
        \qquad\text{and}\qquad
        a_1 - \mu = \frac{\mu(1-p) - \mu(1-\varepsilon)}{1-\varepsilon}
        = -\frac{\varepsilon(1/2 - \mu)}{1-\varepsilon}.
    \end{equation}
    Since $U_{C_e}(x)=\mu$ for every $x\in A(C_e)$, the
definition in~\eqref{eq:divergences} gives
\[
\chi^2(V_e\|U_{C_e})
=
\sum_{x\in A(C_e)}\frac{(V_e(x)-\mu)^2}{\mu}.
\]
We split this sum into two groups. The $s-2$ oriented edges
in $A(C_e)\setminus A(e)$ each contribute
$(a_0-\mu)^2/\mu$, while the two orientations of $e$,
which form $A(e)$, each contribute $(a_1-\mu)^2/\mu$.
Thus,
\[
\chi^2(V_e\|U_{C_e})
=
\frac{(s-2)(a_0-\mu)^2+2(a_1-\mu)^2}{\mu}.
\]
Substituting the expressions for $a_0-\mu$ and $a_1-\mu$
from~\eqref{eq:a_0_a_1_}, we obtain
    \begin{align}
        \chi^2(V_e \,\|\, U_{C_e})
        &=\sum_{x \in A(C_e)} \frac{\big(V_e(x) - \mu\big)^2}{\mu}\nonumber\\
        &=\frac{1}{\mu}\Big( (s-2)\frac{\mu^2\varepsilon^2}{(1-\varepsilon)^2}
        + 2\frac{\varepsilon^2(1/2-\mu)^2}{(1-\varepsilon)^2} \Big)
        = \frac{\varepsilon^2}{\mu(1-\varepsilon)^2}
        \Big( \mu - 2\mu^2 + \frac12 - 2\mu + 2\mu^2 \Big)\nonumber\\
        &= \frac{\varepsilon^2(1/2 - \mu)}{\mu(1-\varepsilon)^2}
        \le \frac{\varepsilon^2}{2\mu(1-\varepsilon)^2},\label{eq:chi_bound}
    \end{align}
    where the third equality uses $\mu s = 1$.  
    Dividing~\eqref{eq:chi_bound} by $\ln 2>0$, we obtain
\[
\frac{\chi^2(V_e\|U_{C_e})}{\ln 2}
\le
\frac{\varepsilon^2}{2\mu(1-\varepsilon)^2\ln 2},
\]
which is the second inequality in~\eqref{eq:D_V_bound}.
Together with the first inequality, already established
using Proposition~1, this completes the proof.
\end{proof}

In what follows, we denote by $D_V$ the common value of
$D(V_e \,\|\, U_{C_e})$ for $e \in E(K)$. The next lemma computes the
rate of the encoder $T_1$ defined in~\eqref{eq:T_1_def}, and shows that, compared to $T_0$, the
rate of $T_1$ differs by the quantity
$\eta$ defined in~\eqref{eq:T_1_rate}. 

\begin{lemma}\label{lemma:T_1_rate}
    The encoder $T_1$ defined in~\eqref{eq:T_1_def} satisfies
    \begin{equation}\label{eq:T_1_rate}
        I(X;T_1) = \log N - \eta,
        \qquad\text{where}\qquad
        \eta = \lambda \log\frac{m+1}{m-1} - \frac{B(1-\varepsilon)}{N}\, D_V.
    \end{equation}
\end{lemma}

\begin{proof}
    Since $X$ is uniform on $A(\G)$, we have that
    $I(X;T_1) = \log M - H(X \mid T_1)$, and, since
    $\log M - \log s = \log N$, the claim~\eqref{eq:T_1_rate} is
    equivalent to prove that
    \begin{equation}\label{eq:T_1_rate_goal}
        H(X \mid T_1) = \log s + \eta.
    \end{equation}
    Moreover, since $H(X \mid T_0) = \log s$ by Lemma~\ref{lemma:T_0}, this implies that $\eta$ is the increase in conditional entropy obtained by
    replacing $T_0$ with $T_1$. We now prove~\eqref{eq:T_1_rate_goal}.

    By Lemma~\ref{lemma:T_1} and~\eqref{eq:conditional_entropy}, and
    since $|E(K)| = B$ and $|E(G) \setminus E(K)| = N - B$, we have that
    \begin{align}\label{eq:H_X_T_1_expl}
        H(X \mid T_1) &= \lambda\, H(U_K)
        + \frac{1-\varepsilon}{N} \sum_{e \in E(K)} H(V_e)
        + \frac{1}{N}\sum_{e\in E(G) \setminus E(K)} H(U_{C_e})\nonumber\\&= \lambda\, H(U_K)
        + \frac{1-\varepsilon}{N} \sum_{e \in E(K)} H(V_e)
        + \frac{N-B}{N}\, H(U_{C_e}).
    \end{align}
    We now substitute the entropies of the posteriors. Since $U_K$ and
    $U_{C_e}$ are uniform on $k(k-1)$ and $s$ oriented edges,
    respectively, we have that $H(U_K) = \log\big(k(k-1)\big)$ and
    $H(U_{C_e}) = \log s$, and, by Lemma~\ref{lemma:perturbation},
    $H(V_e) = \log s - D_V$. Hence,~\eqref{eq:H_X_T_1_expl} becomes
    \begin{equation}\label{eq:H_X_T_1_final}
        H(X \mid T_1) = \lambda \log\big(k(k-1)\big)
        + \frac{B(1-\varepsilon)}{N}\big(\log s - D_V\big)
        + \frac{N-B}{N} \log s.
    \end{equation}
    Finally, by Lemma~\ref{lemma:T_1}, the probabilities of the values of
    $T_1$ sum to $1$, that is,
    \begin{equation}\label{eq:T_1_weights}
        \lambda + \frac{B(1-\varepsilon)}{N} + \frac{N-B}{N} = 1.
    \end{equation}
    Multiplying~\eqref{eq:T_1_weights} by $\log s$ and subtracting the obtained 
    result from~\eqref{eq:H_X_T_1_final}, the terms
    $\frac{B(1-\varepsilon)}{N}\log s$ and $\frac{N-B}{N}\log s$ cancel, and
    we obtain that
    \begin{equation*}
        H(X \mid T_1) - \log s = \lambda \log\frac{k(k-1)}{s}
        - \frac{B(1-\varepsilon)}{N}\, D_V = \eta,
    \end{equation*}
    where the last equality follows since $\frac{k(k-1)}{s} = \frac{(m+1)m}{m(m-1)} = \frac{m+1}{m-1}$.
    This proves~\eqref{eq:T_1_rate_goal}.
\end{proof}

\section{Hardness of the information bottleneck}\label{sec:IB}

In this section, we prove our results for the information bottleneck,
which are formalized in the following theorems.

\begin{theorem}\label{thm:IBL_hard}
    \IBL{} is NP-hard. Moreover, unless $\mathrm{P} = \mathrm{NP}$, there
    is no polynomial-time algorithm that, given a rational joint
    distribution $P_{XY}$ on a finite alphabet $\cX \times \cY$, outputs a
    value $v$ such that
    \begin{equation}\label{eq:IBL_approx}
        \big| v - L_{1/2}(P_{XY}) \big| \le \frac{1}{54\,|\cX|^{3/2}}.
    \end{equation}
\end{theorem}

\begin{theorem}\label{thm:IBT_hard}
    \IBT{} is NP-hard. Moreover, unless $\mathrm{P} = \mathrm{NP}$, there
    is no polynomial-time algorithm that, given a rational joint
    distribution $P_{XY}$ on a finite alphabet $\cX \times \cY$ and a
    non-negative rational number $R$, outputs a value $v$ such that
    \begin{equation}\label{eq:IBT_approx}
        \big| v - \IB(P_{XY},R) \big| \le \frac{1}{144\,|\cX|^{3/2}}.
    \end{equation}
\end{theorem}

The rest of this section is organized as follows. In Section~\ref{subsec:IB_costs}, given a channel $W$ from $\cX$ to $\cY$, we define the cost
of a probability distribution $Q$ as $g_W(Q) = H(WQ) - \frac12 H(Q)$, and
we show that the objective of the Lagrangian depends on an encoder only
through the average cost of its posteriors. In
Section~\ref{subsec:IB_turan}, we introduce the channel that maps each
oriented edge of a graph to one of its endpoints, and we provide, for
this channel, a lower bound on the cost of every distribution on the
oriented edges of the graph in terms of its clique number. In Section~\ref{subsec:IB_reduction}, we define
the joint distribution of the reduction, and we show that, if $G$ has
no $k$-clique, the encoder $T_0$ is optimal, while, if $G$ contains a
$k$-clique, the encoder $T_1$ strictly improves on it. Successively, we use these
results to prove Theorem~\ref{thm:IBL_hard} in
Section~\ref{subsec:IBL_proof}. In Section~\ref{subsec:other_beta}, we
extend Theorem~\ref{thm:IBL_hard} to every fixed multiplier
$\beta \in (0,1)$, and, in Section~\ref{subsec:IBT_proof}, we prove
Theorem~\ref{thm:IBT_hard}. Finally, in Section~\ref{subsec:padding}, we
improve the exponent $3/2$ in Theorems~\ref{thm:IBL_hard}
and~\ref{thm:IBT_hard} and in Proposition~\ref{prop:other_beta} to
$1 + \alpha$, for every fixed $\alpha > 0$ (Theorem~\ref{thm:padding}). We
made no attempt to optimize the constants in these results.

\subsection{The Lagrangian in terms of the posteriors}\label{subsec:IB_costs}

We first show that the objective of $L_{1/2}(P_{XY})$ depends on an
encoder only through its posteriors. Given a channel $W$ from $\cX$ to
$\cY$, we define the \emph{cost} of a probability distribution $Q$ on
$\cX$ as
\begin{equation}\label{eq:cost_def}
    g_W(Q) = H(WQ) - \frac12 H(Q).
\end{equation}
Using~\eqref{eq:cost_def}, we can express the objective of the
Lagrangian in terms of the posteriors of the encoder $T$, as shown in
the following lemma.
\begin{lemma}\label{lemma:lagrangian_costs}
    Let $W$ be the channel from $\cX$ to $\cY$ induced by the joint distribution $P_{XY}$, and let $T$ be an encoder
    of $X$. Then, it holds that
    \begin{equation}\label{eq:lagrangian_costs}
        I(Y;T) - \frac12 I(X;T) = g_W(P_X)
        - \sum_{t \in \cT} P_T(t)\, g_W\big(P_{X\mid T}(\cdot \mid t)\big).
    \end{equation}
    In particular,
    $L_{1/2}(P_{XY}) = g_W(P_X) - \inf_T \sum_{t \in \cT} P_T(t)\, g_W\big(P_{X\mid T}(\cdot \mid t)\big)$.
\end{lemma}

\begin{proof}
    By~\eqref{eq:MI} and~\eqref{eq:conditional_entropy}, we have that
    \begin{equation}\label{eq:I_X_T_post}
        I(X;T) = H(P_X) - \sum_{t \in \cT} P_T(t)\, H\big(P_{X\mid T}(\cdot \mid t)\big)
    \end{equation}
    and, similarly,
    \begin{equation}\label{eq:I_Y_T_post}
        I(Y;T) = H(P_Y) - \sum_{t \in \cT} P_T(t)\, H\big(P_{Y\mid T}(\cdot \mid t)\big).
    \end{equation}
We first rewrite~\eqref{eq:I_Y_T_post} in terms of the
channel $W$. Since $P_Y=WP_X$ and, by
Lemma~\ref{lemma:posterior_Y_given_T},
$P_{Y\mid T}(\cdot\mid t)=WP_{X\mid T}(\cdot\mid t)$,
we have
\[
I(Y;T)
=
H(WP_X)
-
\sum_{t\in\cT}P_T(t)\,
H\bigl(WP_{X\mid T}(\cdot\mid t)\bigr).
\]
Subtracting one half of~\eqref{eq:I_X_T_post} and grouping
the unconditional and conditional entropy terms separately,
we obtain
\begin{align*}
I(Y;T)-\frac12 I(X;T)
&=
\left[H(WP_X)-\frac12 H(P_X)\right]
\\
&\quad-
\sum_{t\in\cT}P_T(t)
\left[
H\bigl(WP_{X\mid T}(\cdot\mid t)\bigr)
-\frac12 H\bigl(P_{X\mid T}(\cdot\mid t)\bigr)
\right]
\\
&=
g_W(P_X)
-
\sum_{t\in\cT}P_T(t)\,
g_W\bigl(P_{X\mid T}(\cdot\mid t)\bigr),
\end{align*}
where the last equality uses the definition
of $g_W$ in~\eqref{eq:cost_def}.
This proves~\eqref{eq:lagrangian_costs}.

Finally, the marginal $P_X$ and the channel $W$ are fixed,
so $g_W(P_X)$ does not depend on the encoder.
Thus, maximizing the Lagrangian objective is equivalent
to minimizing the average cost of the posterior distributions.
Taking the supremum over all encoders in
\eqref{eq:lagrangian_costs} therefore gives
\[
L_{1/2}(P_{XY})
=
g_W(P_X)
-
\inf_T
\sum_{t\in\cT}P_T(t)\,
g_W\bigl(P_{X\mid T}(\cdot\mid t)\bigr),
\]
which proves the last claim.
\end{proof}
By Lemma~\ref{lemma:lagrangian_costs}, the objective of the Lagrangian
depends on the encoder $T$ only through the average cost of its
posteriors. Moreover,
Lemma~\ref{lemma:lagrangian_costs} also provides a natural strategy to lower
and upper bound $L_{1/2}(P_{XY})$. Indeed, if $g_W(Q) \ge c$ for every
probability distribution $Q$ on $\cX$, then, since each posterior
$P_{X\mid T}(\cdot \mid t)$ is a probability distribution on $\cX$, the
infimum in Lemma~\ref{lemma:lagrangian_costs} is at least $c$, and hence
\begin{equation}\label{eq:L_upper_c}
    L_{1/2}(P_{XY}) \le g_W(P_X) - c.
\end{equation}
On the other hand, for every fixed encoder $T$ of $X$,
Lemma~\ref{lemma:lagrangian_costs} yields the following lower bound:
\begin{equation}\label{eq:L_lower_T}
    L_{1/2}(P_{XY}) \ge g_W(P_X) - \sum_{t \in \cT} P_T(t)\, g_W\big(P_{X\mid T}(\cdot \mid t)\big).
\end{equation}

\subsection{The endpoint channel and an entropic Tur\'an bound}\label{subsec:IB_turan}

Let $G$ be a graph with at least one edge. We consider the channel $W_G$
from $A(G)$ to $V(G)$ that maps each oriented edge of $G$ to one of its two
endpoints, chosen uniformly at random, that is,
\begin{equation}\label{eq:W_G}
    W_G\big(w \mid (u,v)\big) =
    \begin{cases}
        1/2, & w \in \{u,v\},\\
        0, & \text{otherwise},
    \end{cases}
\end{equation}
for every $(u,v) \in A(G)$ and $w \in V(G)$. In what follows, we write $g_G$ for the cost $g_{W_G}$ associated with
the channel $W_G$, as defined in~\eqref{eq:cost_def}, that is, we set
$g_G(Q) = H(W_G Q) - \frac12 H(Q)$ for every probability distribution
$Q$ on $A(G)$. For every integer $j \ge 2$, let
\begin{equation}\label{eq:c_j}
    c_j = -\frac12 \log\Big(1 - \frac1j\Big),
\end{equation}
and note that $c_j$ is strictly decreasing in $j$.

The next lemma computes the cost of the uniform distribution on the
oriented edges of a clique, which, by Lemma~\ref{lemma:T_0}, is the form
of the posteriors of $T_0$.
\begin{lemma}\label{lemma:cost_clique}
    Let $C$ be a clique of $G$ of size $j \ge 2$. Then $W_G U_C$ is the
    uniform distribution on $C$, and $g_G(U_C) = c_j$.
\end{lemma}

\begin{proof}
    We first compute $W_G U_C$. For every $v \in V(G)$,
    by~\eqref{eq:U_C} and~\eqref{eq:W_G}  we have that
    \begin{equation*}
        (W_G U_C)(v) = \frac{1}{j(j-1)} \sum_{(w,u) \in A(C)} W_G\big(v \mid (w,u)\big),
    \end{equation*}
    where each term of the sum is $1/2$ if $v \in \{w,u\}$, and $0$
    otherwise. Therefore, if $v \in C$, exactly $2(j-1)$ oriented edges of $A(C)$
    have $v$ as an endpoint, so that
    $(W_G U_C)(v) = \frac{1}{j(j-1)} \cdot 2(j-1) \cdot \frac12 = \frac1j$;
    if $v \notin C$, no oriented edge of $A(C)$ has $v$ as an endpoint, so
    that $(W_G U_C)(v) = 0$. Hence, $W_G U_C$ is uniform on $C$, and
    $H(W_G U_C) = \log j$. Moreover, since $H(U_C) = \log\big(j(j-1)\big)$, we
    finally get that
    \begin{equation*}
        g_G(U_C) = \log j - \frac12 \log\big(j(j-1)\big)
        = -\frac12 \log\Big(1 - \frac1j\Big) = c_j,
    \end{equation*}
    which concludes the proof.
\end{proof}

We now show that no probability distribution $Q$ on $A(G)$ has a cost $g_G(Q)$ smaller
than $c_{\omega(G)}$,  
where $\omega(G)$ is the clique number of $G$ and $c_j$ is as in \eqref{eq:c_j}. 
To this end, we first recall the following result of
Motzkin and Straus \cite{MS65}.

\begin{lemma}[\cite{MS65}]\label{lemma:motzkin_straus}
    For every probability distribution $P$ on $V(G)$,
    \begin{equation}\label{eq:motzkin_straus}
        \sum_{(u,v) \in A(G)} P(u)\, P(v) \le 1 - \frac{1}{\omega(G)}.
    \end{equation}
\end{lemma}
Lemma~\ref{lemma:motzkin_straus} allows us to prove the following useful result.

\begin{lemma}\label{lemma:entropic_turan}
    For every probability distribution $Q$ on $A(G)$, it holds that
    \begin{equation}\label{eq:entropic_turan}
        g_G(Q) \ge c_{\omega(G)},
    \end{equation}
    with equality if $Q = U_C$ for a clique $C$ of size $\omega(G)$.
\end{lemma}

\begin{proof}
    The idea of the proof is to express $g_G(Q)$ as a relative entropy,
    and then to bound it from below by Jensen's inequality and
    Lemma~\ref{lemma:motzkin_straus}.

    Let $\pi = W_G Q$, and let $Q_1$ and $Q_2$ be the distributions on
    $V(G)$ defined by
    \begin{equation*}
        Q_1(v) = \sum_{u : (v,u) \in A(G)} Q(v,u)
        \qquad\text{and}\qquad
        Q_2(v) = \sum_{u : (u,v) \in A(G)} Q(u,v).
    \end{equation*}
    Since, by~\eqref{eq:W_G}, an oriented edge is mapped to each of its
    two endpoints with probability $1/2$, we have that
    \begin{align}\label{eq:pi_Q1_Q2}
        \pi(v) &=\sum_{(w,u)\in A(G)}W_G(v\mid (w, u))Q(w,u) \nonumber\\&=\frac12 \big( Q_1(v) + Q_2(v) \big),
        \qquad \forall v \in V(G).
    \end{align}
    We extend $Q$ to a distribution on $V(G) \times V(G)$ by setting it to
    zero outside $A(G)$. We notice that if $Q(u,v) > 0$, then
    $\pi(u) \ge \frac12 Q_1(u) > 0$ and $\pi(v) \ge \frac12 Q_2(v) > 0$.
    Hence $\supp(Q) \subseteq \supp(\pi \otimes \pi)$, and we can write
    \begin{align}
        D(Q \,\|\, \pi \otimes \pi)
        &= -H(Q) - \sum_{(u,v) \in A(G)} Q(u,v) \log \pi(u)
        - \sum_{(u,v) \in A(G)} Q(u,v) \log \pi(v) \nonumber\\
        &= -H(Q) - \sum_{v \in V(G)} \big( Q_1(v) + Q_2(v) \big) \log \pi(v) \nonumber\\
        &= -H(Q) - 2 \sum_{v \in V(G)} \pi(v) \log \pi(v)
        = 2H(\pi) - H(Q) = 2 g_G(Q), \label{eq:D_equals_2g}
    \end{align}
    where the second equality follows from the definitions of $Q_1$ and
    $Q_2$, and the third one from~\eqref{eq:pi_Q1_Q2}.

    We now bound $D(Q \,\|\, \pi \otimes \pi)$ from below. By Jensen's
    inequality applied to the concave function $\log$, and by
    Lemma~\ref{lemma:motzkin_straus}, we have that
    \begin{align}
        D(Q \,\|\, \pi \otimes \pi)
        &= -\sum_{(u,v) \in \supp(Q)} Q(u,v) \log\frac{\pi(u)\pi(v)}{Q(u,v)}
        \ge -\log\Big( \sum_{(u,v) \in A(G)} \pi(u)\pi(v) \Big)\nonumber\\
        &\ge -\log\Big(1 - \frac{1}{\omega(G)}\Big) = 2c_{\omega(G)}.\label{eq:last_step_turan}
    \end{align}
    Combining~\eqref{eq:last_step_turan} with~\eqref{eq:D_equals_2g}
    proves~\eqref{eq:entropic_turan}. The equality case, instead, follows from
    Lemma~\ref{lemma:cost_clique}, since for a clique $C$ of size $\omega(G)$ we have that $g_G(U_C)=c_{\omega(G)}$.
\end{proof}

The bound of Lemma~\ref{lemma:entropic_turan} can be seen as an entropic
form of Tur\'an's theorem. For symmetric distributions, that is, when
$Q(u,v) = Q(v,u)$ for all $(u,v) \in A(G)$, it was obtained by Chao and
Yu~\cite{CY26}. Finally, we bound the decrease of $c_j$ when $j$
increases by one.

\begin{lemma}\label{lemma:c_gap}
    For every integer $m \ge 2$,
    $c_m - c_{m+1} \ge \frac{1}{2m^2 \ln 2}$.
\end{lemma}

\begin{proof}
    By~\eqref{eq:c_j},
    $c_m - c_{m+1} = \frac12 \log\big(\frac{m}{m-1} \cdot \frac{m}{m+1}\big)
    = -\frac12 \log\big(1 - \frac{1}{m^2}\big) \ge \frac{1}{2m^2 \ln 2}$,
    where the inequality follows from $-\ln(1-z) \ge z$ for
    $0 \le z < 1$.
\end{proof}

\subsection{The reduction}\label{subsec:IB_reduction}

Let $(G,k)$ be an instance of \clique{} satisfying~\eqref{eq:assumption_on_G},
and let $\G$ be the extended graph of Section~\ref{subsec:extended_graph}.
We let $X$ be uniformly distributed on $\cX = A(\G)$, as
in~\eqref{eq:P_X}, and, given that $X = (u,v)$, we let the value of the random variable $Y$ be one of the two
endpoints of $(u,v)$, chosen uniformly at random. Hence $\cY = V(\G)$,
the channel from $X$ to $Y$ is $W_{\G}$, defined in~\eqref{eq:W_G}, and the joint distribution of $(X,Y)$ is
\begin{equation}\label{eq:IB_joint}
    P_{XY}\big((u,v),w\big) =
    \begin{cases}
        \dfrac{1}{2M}, & \text{if }w \in \{u,v\},\\[6pt]
        0, & \text{otherwise},
    \end{cases}
    \qquad \forall (u,v) \in A(\G),\ w \in V(\G).
\end{equation}
Let
\begin{equation}\label{eq:Lambda_0}
    \Pi_0 = g_{\G}(P_X) - c_m.
\end{equation}

The next lemma shows that the encoder $T_0$ achieves exactly the value $\Pi_0$ defined in~\eqref{eq:Lambda_0}.
\begin{lemma}\label{lemma:IB_T0}
    The encoder $T_0$ defined in~\eqref{eq:T_0_def} satisfies
    $I(Y;T_0) - \frac12 I(X;T_0) = \Pi_0$. In particular,
    $L_{1/2}(P_{XY}) \ge \Pi_0$.
\end{lemma}

\begin{proof}
    By Lemma~\ref{lemma:T_0}, the posteriors of $T_0$ are the
    distributions $U_{C_e}$, with probability $1/N$ each. Hence, by
    Lemma~\ref{lemma:lagrangian_costs}, we have that
    \begin{equation}\label{eq:IB_T0_costs}
        I(Y;T_0) - \frac12 I(X;T_0)
        = g_{\G}(P_X) - \frac1N \sum_{e \in E(G)} g_{\G}(U_{C_e}).
    \end{equation}
    Since each $C_e$ is a clique of $\G$ of size $m$,
    Lemma~\ref{lemma:cost_clique} gives $g_{\G}(U_{C_e}) = c_m$, and
    hence~\eqref{eq:IB_T0_costs} equals $g_{\G}(P_X) - c_m = \Pi_0$.
    Finally, $L_{1/2}(P_{XY}) \ge \Pi_0$ follows from the
    definition~\eqref{eq:L_def} of $L_{1/2}(P_{XY})$ as a supremum over all
    encoders.
\end{proof}

We now show that, if $G$ has no $k$-clique, then the encoder $T_0$ is optimal, 
that is, $L_{1/2}(P_{XY})=\Pi_0$. While,
if $G$ contains a $k$-clique, then the encoder $T_1$ strictly improves on $T_0$, that is, $L_{1/2}(P_{XY})>\Pi_0$.
This is the conceptual step to solve the \clique\ problem through the estimation  of $L_{1/2}(P_{XY})$, and it will 
prove the hardness of \IBL{}.

\subsubsection{No-instances: $G$ does not contain a $k$-clique}
We demonstrate that, if $G$ does not contain a $k$-clique, then no encoder of $X$ improves on
$T_0$.
\begin{lemma}\label{lemma:IB_no}
    Assume that $\omega(G) < k$. Then every encoder $T$ of $X$ satisfies
    \begin{equation}\label{eq:IB_no}
        I(Y;T) - \frac12 I(X;T) \le \Pi_0.
    \end{equation}
    Consequently, $L_{1/2}(P_{XY}) = \Pi_0$, and
    $\IB(P_{XY},R) \le \frac{R}{2} + \Pi_0$ for every $R \ge 0$.
\end{lemma}

\begin{proof}
    Since $\omega(G) \le k - 1 = m$, property~3 of Lemma~\ref{lemma:gadget}
    gives $\omega(\G) = m$. Hence, by Lemma~\ref{lemma:entropic_turan}
    applied to $\G$, we have that $g_{\G}(Q) \ge c_m$ for every probability
    distribution $Q$ on $A(\G)$, and in particular this holds for every posterior of
    an encoder $T$ of $X$. Thus, by Lemma~\ref{lemma:lagrangian_costs},
    every encoder $T$ of $X$ satisfies
    \begin{equation*}
        I(Y;T) - \frac12 I(X;T)
        = g_{\G}(P_X) - \sum_{t \in \cT} P_T(t)\, g_{\G}\big(P_{X\mid T}(\cdot \mid t)\big)
        \le g_{\G}(P_X) - c_m = \Pi_0,
    \end{equation*}
    which proves~\eqref{eq:IB_no}. Taking the supremum over all encoders
    gives $L_{1/2}(P_{XY}) \le \Pi_0$, and, together with
    Lemma~\ref{lemma:IB_T0}, $L_{1/2}(P_{XY}) = \Pi_0$. Finally, if
    $I(X;T) \le R$, then~\eqref{eq:IB_no} gives
    $I(Y;T) \le \frac12 I(X;T) + \Pi_0 \le \frac{R}{2} + \Pi_0$,
    and taking the supremum over all such encoders yields
    $\IB(P_{XY},R) \le \frac{R}{2} + \Pi_0$.
\end{proof}

\subsubsection{Yes-instances: $G$ does contain a $k$-clique}

We now assume that $G$ contains a $k$-clique $K$, and we show that the
encoder $T_1$ of Section~\ref{subsec:T_1} strictly improves on $T_0$. To this end, we choose
\begin{equation}\label{eq:IB_p}
    p = \frac{1}{4m^2},
    \qquad\text{so that}\qquad
    \varepsilon = 2\mu p = \frac{\mu}{2m^2}
    \quad\text{and}\quad
    \lambda = \frac{B\varepsilon}{N} = \frac{B\mu}{2Nm^2}.
\end{equation}

By Lemma~\ref{lemma:T_1}, the encoder $T_1$ replaces part of the
posteriors $U_{C_e}$ of $T_0$, whose cost is $c_m$, with the posterior
$U_K$, whose cost is $c_k < c_m$. In doing so, however, it perturbs
$U_{C_e}$ into $V_e$ for every $e \in E(K)$. The next lemma shows that
this perturbation increases the cost by at most a quantity quadratic in
$\varepsilon$.

\begin{lemma}\label{lemma:IB_Ve_cost}
    For every $e \in E(K)$,
    \begin{equation}\label{eq:IB_Ve_cost}
        g_{\bar G}(V_e)-c_m
\le
\frac12 D(V_e\|U_{C_e})
\le
\frac{\varepsilon^2}{4\mu(1-\varepsilon)^2\ln 2}.
    \end{equation}
\end{lemma}

\begin{proof}

The distribution $W_{\bar G}U_{C_e}$ is uniform on $C_e$,
whereas $W_{\bar G}V_e$ is supported on the same set.
Thus, replacing $U_{C_e}$ by $V_e$ cannot increase the
output entropy. Any increase in cost is therefore bounded
by one half of the decrease in input entropy, which,
by Lemma~\ref{lemma:perturbation}, equals $D(V_e\|U_{C_e})$.
We now make this comparison explicit.
    By~\eqref{eq:c_j}, we can write
    \begin{equation}\label{eq:c_m_mu}
        c_m = \frac12 \log\frac{m}{m-1} = \log m - \frac12 \log s.
    \end{equation}
    Combining~\eqref{eq:c_m_mu} with the definition~\eqref{eq:cost_def}
    of $g_{\G}$, and using $H(V_e) = \log s - D(V_e\|U_{C_e})$ from
    Lemma~\ref{lemma:perturbation}, we obtain that
    \begin{align}\label{eq:g_Ve_split}
        g_{\G}(V_e) - c_m &= \big( H(W_{\G} V_e) - \log m \big)
        + \frac12 \big( \log s - H(V_e) \big)\nonumber\\
        &= \big( H(W_{\G} V_e) - \log m \big) + \frac12 D(V_e\|U_{C_e}).
    \end{align}
    Since $\supp(V_e) \subseteq A(C_e)$, by~\eqref{eq:W_G} the
    distribution $W_{\G} V_e$ is supported on the $m$ vertices of $C_e$. This implies that $H(W_{\G} V_e) \le \log m$. Hence, the first term
    in~\eqref{eq:g_Ve_split} is non-positive, which proves the first
    inequality in~\eqref{eq:IB_Ve_cost}. The second one immediately follows
    from~\eqref{eq:D_V_bound} of Lemma~\ref{lemma:perturbation}.
\end{proof}

We now show that the encoder $T_1$ strictly improves on $T_0$.
\begin{lemma}\label{lemma:IB_yes}
    Assume that $\omega(G) \ge k$, let $T_1$ be the encoder defined
    in~\eqref{eq:T_1_def} with $p$ as in~\eqref{eq:IB_p}, and let
    \begin{equation}\label{eq:delta_def}
        \delta = \frac{1}{8Nm^4}.
    \end{equation}
    Then $I(Y;T_1) - \frac12 I(X;T_1) > \Pi_0 + \delta$. Consequently,
    $L_{1/2}(P_{XY}) > \Pi_0 + \delta$.
\end{lemma}

\begin{proof}
    The idea of the proof is to compare the average cost of the posteriors
    of $T_1$ with the cost $c_m$ of the posteriors of $T_0$. 
    The encoder $T_0$ has average posterior cost $c_m$.
By Lemma~\ref{lemma:lagrangian_costs}, to show that $T_1$ improves the Lagrangian
objective over $T_0$ by more than $\delta$, it suffices
to prove that its average posterior cost is less than
$c_m-\delta$.
Let
    \begin{equation*}
        \Psi = \sum_{t} P_{T_1}(t)\, g_{\G}\big(P_{X\mid T_1}(\cdot \mid t)\big)
    \end{equation*}
    be the average cost of the posteriors of $T_1$. By
    Lemma~\ref{lemma:T_1}, and since $g_{\G}(U_K) = c_k$ and
    $g_{\G}(U_{C_e}) = c_m$ by Lemma~\ref{lemma:cost_clique}, we have that
    \begin{equation*}
        \Psi = \lambda c_k + \frac{1-\varepsilon}{N} \sum_{e \in E(K)} g_{\G}(V_e)
        + \frac{N-B}{N}\, c_m.
    \end{equation*}
    Since the probabilities in~\eqref{eq:T_1_distribution} sum to $1$, we
    can write $c_m$ as
    \begin{equation}\label{eq:c_m_weights}
        c_m = \lambda c_m + \frac{1-\varepsilon}{N} \sum_{e \in E(K)} c_m
        + \frac{N-B}{N}\, c_m.
    \end{equation}
    Subtracting~\eqref{eq:c_m_weights} from $\Psi$, the terms
    corresponding to $e \in E(G) \setminus E(K)$ cancel out, and we obtain
    that
    \begin{equation}\label{eq:Psi_minus_cm}
        \Psi - c_m = -\lambda (c_m - c_k)
        + \frac{1-\varepsilon}{N} \sum_{e \in E(K)} \big( g_{\G}(V_e) - c_m \big).
    \end{equation}
The negative term $-\lambda(c_m-c_k)$ represents the
decrease in average posterior cost contributed by $U_K$,
relative to the baseline $c_m$.
The second term is the total weighted change in cost
caused by perturbing $U_{C_e}$ into $V_e$.
We bound $\lambda(c_m-c_k)$ from below and the second
term from above.

We first bound \(\lambda(c_m-c_k)\) from below.
Since $k = m+1$, Lemma~\ref{lemma:c_gap} gives
    \begin{equation}\label{eq:IB_gain}
        \lambda(c_m - c_k) \ge \frac{\lambda}{2m^2 \ln 2}.
    \end{equation}
    We now bound the perturbation term from above. We first note that $\varepsilon$ is small. Indeed,
    since $m \ge 3$, we have that $\mu \le \frac16$ and $m^2 \ge 9$, so
    that $\varepsilon = \frac{\mu}{2m^2} \le \frac{1}{108}$, and hence
    $\frac{1}{1-\varepsilon} \le \frac{108}{107}$. Therefore, since
    $|E(K)| = B$, by Lemma~\ref{lemma:IB_Ve_cost} and using
    $\lambda = \frac{B\varepsilon}{N}$ and
    $\frac{\varepsilon}{\mu} = \frac{1}{2m^2}$, we get that
    \begin{equation}\label{eq:IB_loss}
        \frac{1-\varepsilon}{N} \sum_{e \in E(K)} \big( g_{\G}(V_e) - c_m \big)
        \le \frac{B\varepsilon^2}{4N\mu(1-\varepsilon)\ln 2}
        = \frac{\lambda}{m^2 \ln 2} \cdot \frac{1}{8(1-\varepsilon)}
        \le \frac{\lambda}{m^2 \ln 2} \cdot \frac{27}{214}.
    \end{equation}
    Substituting~\eqref{eq:IB_gain} and~\eqref{eq:IB_loss}
    into~\eqref{eq:Psi_minus_cm}, we get that
    \begin{equation*}
        \Psi - c_m \le -\frac{\lambda}{m^2 \ln 2}\Big( \frac12 - \frac{27}{214} \Big)
        = -\frac{40}{107} \cdot \frac{\lambda}{m^2 \ln 2}
        < -\frac{\lambda}{2m^2},
    \end{equation*}
    where the last inequality follows from $\ln 2 < \frac{7}{10}$, which
    gives $\frac{40}{107 \ln 2} > \frac{400}{749} > \frac12$. Since
    $\lambda = \frac{B\mu}{2Nm^2}$ by~\eqref{eq:IB_p} and
    $B\mu = \frac{k(k-1)}{2} \cdot \frac{1}{(k-1)(k-2)} = \frac{k}{2(k-2)} \ge \frac12$,
    we obtain that
    \begin{equation}\label{eq:Psi_bound}
        \Psi - c_m < -\frac{B\mu}{4Nm^4} \le -\frac{1}{8Nm^4} = -\delta.
    \end{equation}
    Finally, by Lemma~\ref{lemma:lagrangian_costs}
    and~\eqref{eq:Psi_bound},
    $I(Y;T_1) - \frac12 I(X;T_1) = g_{\G}(P_X) - \Psi > g_{\G}(P_X) - c_m + \delta = \Pi_0 + \delta$. The last statement $L_{1/2}(P_{XY})> \Pi_0 + \delta$ follows from the definition of $L_{1/2}(P_{XY})$ as a supremum over all encoders.
\end{proof}

\subsection{Proof of Theorem~\ref{thm:IBL_hard}}\label{subsec:IBL_proof}
We now combine the results of the previous sections to prove Theorem~\ref{thm:IBL_hard}. Let $(G,k)$ be an
instance of \clique{} satisfying~\eqref{eq:assumption_on_G}, and let
$P_{XY}$ be the joint distribution defined in~\eqref{eq:IB_joint}.
By Lemmas~\ref{lemma:IB_no} and~\ref{lemma:IB_yes}, we have that
$L_{1/2}(P_{XY}) = \Pi_0$ if $\omega(G) < k$, and
$L_{1/2}(P_{XY}) > \Pi_0 + \delta$ if $\omega(G) \ge k$. Hence, the
two cases are separated by a gap of size $\delta$, and we place the
threshold in the middle of this gap. However, since $\Pi_0$ involves
logarithms, and is therefore in general irrational, we choose the
threshold $\tau$ as a rational number such that
\begin{equation}\label{eq:tau_def}
    \Big| \tau - \Big( \Pi_0 + \frac{\delta}{2} \Big) \Big| \le \frac{\delta}{12}.
\end{equation}
Therefore, the instance of \IBL{} associated with $(G,k)$ is the pair $(P_{XY},\tau)$, where $P_{XY}$ is defined in~\eqref{eq:IB_joint} and $\tau$ in~\eqref{eq:tau_def}.
Before presenting the proof, since the inapproximability results are stated in terms of the size of
the input, the next lemma relates $\delta$ defined in~\eqref{eq:delta_def} to the size $M = |\cX|$ of
the alphabet of $X$.

\begin{lemma}\label{lemma:delta_bound}
    Let $\delta$ be as in~\eqref{eq:delta_def}. Then
    $\delta \ge \frac{1}{18M^{3/2}}$.
\end{lemma}

\begin{proof}
    We first express $\delta$ in terms of $M$ and $m$. By property~2 of
    Lemma~\ref{lemma:gadget}, $N = M/(m(m-1))$, and substituting this
    into~\eqref{eq:delta_def} we obtain that
    \begin{equation}\label{eq:delta_M}
        \delta = \frac{m(m-1)}{8Mm^4} = \frac{1 - 1/m}{8Mm^2}.
    \end{equation}
    We now bound $m$ from above in terms of $M$. Since
    $B = \frac{m(m+1)}{2}$ and $N \ge B$ by~\eqref{eq:assumption_on_G},
    we have that $M = Nm(m-1) \ge \frac{m^2(m^2-1)}{2}$. Since $m \ge 3$,
    we have that $m^2 - 1 \ge \frac89 m^2$, and hence
    $M \ge \frac49 m^4$, that is,
    \begin{equation}\label{eq:m_upper}
        m^2 \le \frac32 \sqrt{M}.
    \end{equation}
    Finally, using~\eqref{eq:m_upper} and $1 - \frac1m \ge \frac23$
    in~\eqref{eq:delta_M}, we conclude that
    $\delta \ge \frac{2/3}{8M \cdot \frac32\sqrt{M}} = \frac{1}{18M^{3/2}}$.
\end{proof}

\begin{proof}[Proof of Theorem~\ref{thm:IBL_hard}]
     We start by showing that \IBL{} is NP-hard. To this end, we show that the construction of the instance $(P_{XY},\tau)$
    described above is a polynomial-time reduction from \clique{} to
    \IBL{}. Given an instance $(G,k)$ of \clique{}
    satisfying~\eqref{eq:assumption_on_G}, the reduction constructs the
    extended graph $\G$, the joint distribution $P_{XY}$ defined
    in~\eqref{eq:IB_joint}, and a rational threshold $\tau$
    satisfying~\eqref{eq:tau_def}, and outputs the instance $(P_{XY},\tau)$ of \IBL{}.

    We first show that the reduction runs in polynomial time. Since $G$ is
    represented by its adjacency matrix, the size of $(G,k)$ is at least
    $n$, so it suffices to show that the construction runs in time
    polynomial in $n$. By property~4 of Lemma~\ref{lemma:gadget}, $\G$ can
    be constructed in polynomial time, with $|V(\G)| \le n + n^3$ and
    $M \le n^4$. Hence, $P_{XY}$ is a table with
    $|A(\G)| \cdot |V(\G)| \le n^4(n + n^3)$ entries, each equal to $0$ or
    to $1/(2M)$ by~\eqref{eq:IB_joint}, which are rational numbers with
    $O(\log n)$ bits. It remains to compute $\tau$. Note that this does
    not require computing $L_{1/2}(P_{XY})$, but only approximating
    $\Pi_0 + \delta/2$. The value $\delta$ in~\eqref{eq:delta_def} is
    a rational number with $O(\log n)$ bits, while $\Pi_0$ is an
    explicit expression involving polynomially many logarithms of rational
    numbers. Indeed, $H(P_X) = \log M$, $c_m = \frac12 \log\frac{m}{m-1}$,
    and, denoting by $d_v$ the degree of a vertex $v$ in $\G$, each vertex
    $v$ is an endpoint of $2d_v$ oriented edges, each of probability
    $1/M$ and mapped to $v$ with probability $1/2$, so that
    $(W_{\G} P_X)(v) = \frac{d_v}{M}$ and
    \begin{equation*}
        H(W_{\G} P_X) = \sum_{v \in V(\G) : d_v > 0} \frac{d_v}{M} \log\frac{M}{d_v}.
    \end{equation*}
    Hence, by Lemma~\ref{lem:log} in Appendix~\ref{app:log} with additive
    error $\delta/12$, a rational $\tau$ satisfying~\eqref{eq:tau_def} can
    be computed in polynomial time. Hence, the whole reduction runs in polynomial time.

    We now show that the reduction is correct, that is, $(P_{XY},\tau)$ is
    a yes-instance of \IBL{} \textit{if and only if} $(G,k)$ is a yes-instance of \clique{}. Recall that $(P_{XY}, \tau)$ is a yes-instance of \IBL{} if
    there exists an encoder $T$ of $X$ such that
    $I(Y;T) - \frac12 I(X;T) \ge \tau$, and that $(G,k)$ is a yes-instance
    of \clique{} if $\omega(G) \ge k$, that is, $G$ contains a clique of size $k$. 
    By~\eqref{eq:tau_def}, the threshold $\tau$ satisfies
    \begin{equation}\label{eq:tau_window}
        \Pi_0 + \frac{5\delta}{12} \le \tau \le \Pi_0 + \frac{7\delta}{12},
    \end{equation}
    that is, $\tau$ lies strictly between $\Pi_0$ and
    $\Pi_0 + \delta$. Assume first that $(G,k)$ is a yes-instance, that is, $\omega(G) \ge k$. By
    Lemma~\ref{lemma:IB_yes} and the right inequality
    in~\eqref{eq:tau_window}, the encoder $T_1$ satisfies
    \begin{equation*}
        I(Y;T_1) - \frac12 I(X;T_1) > \Pi_0 + \delta
        \ge \tau - \frac{7\delta}{12} + \delta = \tau + \frac{5\delta}{12}.
    \end{equation*}
    Hence, $(P_{XY},\tau)$ is a yes-instance of \IBL{}, and
    $L_{1/2}(P_{XY}) > \tau + \frac{5\delta}{12}$. Assume now that $(G,k)$ is a no-instance, that is,
    $\omega(G) < k$. By Lemma~\ref{lemma:IB_no} and the left inequality
    in~\eqref{eq:tau_window}, every encoder $T$ of $X$ satisfies
    $I(Y;T) - \frac12 I(X;T) \le \Pi_0 \le \tau - \frac{5\delta}{12} < \tau$.
    Hence, $(P_{XY},\tau)$ is a no-instance of \IBL{}, and
    $L_{1/2}(P_{XY}) \le \tau - \frac{5\delta}{12}$. This proves that the
    reduction is correct, and hence that \IBL{} is NP-hard.

    It remains to prove the inapproximability result. By the two cases
    above, $L_{1/2}(P_{XY})$ lies above $\tau + \frac{5\delta}{12}$ if
    $\omega(G) \ge k$, and below $\tau - \frac{5\delta}{12}$ if $\omega(G)<k$.
    Suppose that a polynomial-time algorithm ${\mathcal Alg}$ outputs, given $P_{XY}$, a
    value $v$ with $|v - L_{1/2}(P_{XY})| \le \frac{1}{54M^{3/2}}$. By
    Lemma~\ref{lemma:delta_bound}, this error is at most
    $\frac{\delta}{3} < \frac{5\delta}{12}$. Hence, if $\omega(G) \ge k$,
    then $v \ge L_{1/2}(P_{XY}) - \frac{\delta}{3} > \tau$, while, if
    $\omega(G) < k$, then $v \le L_{1/2}(P_{XY}) + \frac{\delta}{3} < \tau$.
    Thus, running the reduction, then the algorithm ${\mathcal Alg}$, and comparing $v$
    with $\tau$ would decide \clique{} in polynomial time, which is
    impossible unless $\mathrm{P} = \mathrm{NP}$.
\end{proof}

\subsection{General fixed Lagrange multipliers}\label{subsec:other_beta}

The hardness result of Theorem~\ref{thm:IBL_hard} can be extended to any
other fixed multiplier $\beta \in (0,1)$. Note that, for $\beta \ge 1$,
the data-processing inequality gives $L_\beta(P_{XY}) = 0$.

\begin{proposition}\label{prop:other_beta}
    Let $\beta \in \mathbb{Q} \cap (0,1)$ be fixed, and let
    $\kappa_\beta = 2\min\{\beta, 1-\beta\}$. Given a rational joint
    distribution $P_{XY}$, one can construct in polynomial time a rational
    joint distribution $\tilde P_{XZ}$ with the same marginal $P_X$ such
    that
    \begin{equation}\label{eq:L_beta_scaling}
        L_\beta(\tilde P_{XZ}) = \kappa_\beta\, L_{1/2}(P_{XY}).
    \end{equation}
    Consequently, \IBL$_\beta$ is NP-hard, and, unless
    $\mathrm{P} = \mathrm{NP}$, no polynomial-time algorithm approximates
    $L_\beta(P_{XY})$ within an additive error of
    $\kappa_\beta/(54|\cX|^{3/2})$.
\end{proposition}

\begin{proof}
    We first prove~\eqref{eq:L_beta_scaling}. Let
    $\rho_\beta = \beta - \frac{\kappa_\beta}{2}$ and
    $e_\beta = 1 - \kappa_\beta - \rho_\beta$. These three numbers form a
    probability distribution. Indeed, they sum to $1$ by definition, and,
    if $\beta \le \frac12$, then $\kappa_\beta = 2\beta$, $\rho_\beta = 0$
    and $e_\beta = 1 - 2\beta \ge 0$, while, if $\beta \ge \frac12$, then
    $\kappa_\beta = 2(1-\beta)$, $\rho_\beta = 2\beta - 1 \ge 0$ and
    $e_\beta = 0$. Let $J$ be a random variable independent of $(X,Y)$,
    taking values $0$, $1$ and $2$ with probabilities $\kappa_\beta$,
    $\rho_\beta$ and $e_\beta$, respectively, and let $Z$ be the random
    variable on $\cZ = (\{0\} \times \cY) \cup (\{1\} \times \cX) \cup \{\bot\}$,
    where $\bot$ is a new symbol, defined by
    \begin{equation}\label{eq:Z_def}
        Z =
        \begin{cases}
            (0,Y), & \text{if } J = 0,\\
            (1,X), & \text{if } J = 1,\\
            \bot, & \text{if } J = 2.
        \end{cases}
    \end{equation}
    In other words, $Z$ reveals $Y$ with probability $\kappa_\beta$,
    reveals $X$ with probability $\rho_\beta$, and reveals nothing with
    probability $e_\beta$. Let $\tilde P_{XZ}$ be the joint distribution of
    $(X,Z)$. Its marginal on $X$ is $P_X$, so that the encoders of $X$ for
    $P_{XY}$ and for $\tilde P_{XZ}$ coincide.

    Let $T$ be an encoder of $X$. Since $J$ is a function of $Z$ and is
    independent of $(X,Y,T)$, the chain rule gives
    \begin{equation}\label{eq:I_Z_T}
        I(Z;T) = I(Z,J;T) = I(J;T) + I(Z;T \mid J)
        = \kappa_\beta\, I(Y;T) + \rho_\beta\, I(X;T),
    \end{equation}
    since $I(J;T) = 0$, and, given $J$, the random variable $Z$ is
    equivalent to $Y$, to $X$, or constant. Subtracting $\beta I(X;T)$
    from both sides of~\eqref{eq:I_Z_T}, and using
    $\beta - \rho_\beta = \kappa_\beta/2$, we get that
    \begin{equation}\label{eq:L_beta_identity}
        I(Z;T) - \beta I(X;T) = \kappa_\beta \Big( I(Y;T) - \frac12 I(X;T) \Big).
    \end{equation}
    Since~\eqref{eq:L_beta_identity} holds for every encoder, taking the
    supremum over all encoders proves~\eqref{eq:L_beta_scaling}.

    We now show that \IBL{}$_\beta$ is NP-hard by reducing \IBL{} to it. Given an instance
    $(P_{XY},\tau)$ of \IBL{}, consider the instance
    $(\tilde P_{XZ}, \kappa_\beta\tau)$ of \IBL$_\beta$.
    By~\eqref{eq:L_beta_identity}, an encoder satisfies
    $I(Y;T) - \frac12 I(X;T) \ge \tau$ if and only if it satisfies
    $I(Z;T) - \beta I(X;T) \ge \kappa_\beta\tau$, so the two instances have
    the same answer. Moreover, $|\cZ| = |\cX| + |\cY| + 1$, and, since
    $\kappa_\beta$, $\rho_\beta$ and $e_\beta$ are fixed rational numbers,
    by~\eqref{eq:Z_def} each entry of $\tilde P_{XZ}$ is a fixed rational
    multiple of an entry of $P_{XY}$ or of $P_X$. Hence, the instance
    $(\tilde P_{XZ}, \kappa_\beta\tau)$ can be computed in polynomial time,
    and, by Theorem~\ref{thm:IBL_hard}, \IBL$_\beta$ is NP-hard.

    Finally, suppose that a polynomial-time algorithm approximates
    $L_\beta$ within an additive error of $\kappa_\beta/(54|\cX|^{3/2})$.
    Running it on $\tilde P_{XZ}$, which has the same input alphabet $\cX$
    as $P_{XY}$, and dividing its output by $\kappa_\beta$, we obtain,
    by~\eqref{eq:L_beta_scaling}, an approximation of $L_{1/2}(P_{XY})$
    within $1/(54|\cX|^{3/2})$, which, by Theorem~\ref{thm:IBL_hard}, is
    impossible unless $\mathrm{P} = \mathrm{NP}$.
\end{proof}

\begin{remark}\label{rem:real_beta}
    If $\beta$ is irrational, then $\kappa_\beta$, $\rho_\beta$ and
    $e_\beta$ may be irrational, and  the joint distribution $\tilde P_{XZ}$, constructed in the proof of Proposition~\ref{prop:other_beta}, is not a rational
    distribution. Nevertheless, we can still extend Proposition~\ref{prop:other_beta},
    with a different constant, to every fixed real $\beta \in (0,1)$ that
    is computable in polynomial time, that is, such that a rational number
    within distance $2^{-r}$ of $\beta$ can be computed in time polynomial
    in $r$. The idea is to apply the construction~\eqref{eq:Z_def} with a
    rational multiplier $b$ close to $\beta$, chosen depending on the
    instance. Replacing $\beta$ with $b$ changes the objective of every
    encoder by $|\beta - b|\, I(X;T) \le |\beta - b| \log|\cX|$, so that
    $|L_\beta(\tilde P_{XZ}) - L_b(\tilde P_{XZ})| \le |\beta - b| \log|\cX|$.
    The instances of Theorem~\ref{thm:IBL_hard} are separated from the
    threshold by a margin $\gamma = \frac{5\delta}{12}$, which becomes
    $\kappa_b \gamma$ after the construction~\eqref{eq:Z_def}. Hence, if
    $|\beta - b| \log|\cX|$ is a small fraction of $\kappa_b \gamma$, the
    threshold $\kappa_b \tau$ still separates yes- and no-instances. By
    Lemma~\ref{lemma:delta_bound}, $\gamma$ is of order $|\cX|^{-3/2}$, so
    the required precision is of order
    $(|\cX|^{3/2} \log|\cX|)^{-1}$, which corresponds to $O(\log|\cX|)$
    bits. Hence, $b$ can be computed in polynomial time, and the arguments
    in the proof of Proposition~\ref{prop:other_beta} show that, for every
    fixed polynomial-time computable real $\beta \in (0,1)$,
    \IBL$_\beta$ is NP-hard, and that $L_\beta$ cannot be approximated in
    polynomial time within an additive error of order $|\cX|^{-3/2}$,
    unless $\mathrm{P} = \mathrm{NP}$.
\end{remark}

\subsection{Proof of Theorem~\ref{thm:IBT_hard}}\label{subsec:IBT_proof}

We now prove Theorem~\ref{thm:IBT_hard}, using the same joint distribution
$P_{XY}$ defined in~\eqref{eq:IB_joint}. However, in this case, an instance of \IBT{} is a
triple $(P_{XY},R,\vartheta)$. Thus, we need to define the budget $R$ and
the threshold $\vartheta$. By Lemma~\ref{lemma:IB_no}, if
$\omega(G) < k$, then $\IB(P_{XY},R) \le \frac{R}{2} + \Pi_0$ for
every $R \ge 0$. On the other hand, if $\omega(G) \ge k$, then, by
Lemma~\ref{lemma:IB_yes},
$I(Y;T_1) > \frac12 I(X;T_1) + \Pi_0 + \delta$. 

We must choose the budget $R$ to satisfy two requirements.
Whenever $G$ contains a $k$-clique, we need
$R\ge I(X;T_1)$ so that $T_1$ satisfies the rate constraint.
We also need $R$ to be sufficiently close to $I(X;T_1)$
so that the upper bound $\frac{R}{2}+\Pi_0$ for the
no-case given by Lemma~\ref{lemma:IB_no} remains below $I(Y;T_1)$ by a gap of order $\delta$.

The difficulty is that the reduction must choose $R$
without knowing whether $G$ contains a $k$-clique,
while $T_1$ is defined only when such a clique exists.
The next lemma addresses this difficulty by giving an
explicit expression, depending only on $N$ and $k$,
that is defined for every instance and equals $I(X;T_1)$
whenever $G$ contains a $k$-clique.
The reduction can therefore use this expression to choose
$R$ from $(G,k)$ alone, without knowing whether such a
clique exists.
\begin{lemma}\label{lemma:R_star}
    Let $p$, $\varepsilon$ and $\lambda$ be as in~\eqref{eq:IB_p}, let
    $a_0$ and $a_1$ be as in~\eqref{eq:V_e_def}, and let
    \begin{equation}\label{eq:R_star}
        R^* = \log N - \lambda \log\frac{m+1}{m-1}
        + \frac{B(1-\varepsilon)}{N}
        \Big( \log s - (s-2)\, a_0 \log\frac{1}{a_0} - 2a_1 \log\frac{1}{a_1} \Big).
    \end{equation}
    Then $R^*$ depends only on $N$ and $k$, and, if $\omega(G) \ge k$, the
    encoder $T_1$ defined in~\eqref{eq:T_1_def}, with $p$ as
    in~\eqref{eq:IB_p}, satisfies $I(X;T_1) = R^*$.
\end{lemma}

\begin{proof}
    We first show that $R^*$ depends only on $N=|E(G)|$ and $k$. Recall that
    $m = k-1$, $B = k(k-1)/2$, $s = m(m-1)$ and $\mu = 1/s$. Hence,
    by~\eqref{eq:V_e_def} and~\eqref{eq:IB_p}, the quantities $p$,
    $\varepsilon$, $\lambda$, $a_0$ and $a_1$, and thus all the quantities
    in~\eqref{eq:R_star}, depend only on $N$ and $k$. Moreover,
    $a_0, a_1 > 0$, so the logarithms are well defined.

Now assume that $\omega(G)\ge k$, and fix a clique $K$
of size $k$ in $G$. Let $T_1$ be the encoder defined
in~\eqref{eq:T_1_def} using this clique, with $p$ chosen as in~\eqref{eq:IB_p}.
We show that $I(X;T_1)=R^*$.
    By
    Lemma~\ref{lemma:T_1_rate}, we have that
    \begin{equation}\label{eq:R_star_proof_1}
        I(X;T_1) = \log N - \lambda \log\frac{m+1}{m-1} + \frac{B(1-\varepsilon)}{N}\, D_V,
    \end{equation}
    and, by Lemma~\ref{lemma:perturbation}, $D_V = \log s - H(V_e)$. Since,
    by~\eqref{eq:V_e_def}, $V_e$ takes the value $a_0$ on $s-2$ oriented
    edges and the value $a_1$ on $2$ oriented edges, we have that
    $H(V_e) = (s-2)\, a_0 \log\frac{1}{a_0} + 2a_1 \log\frac{1}{a_1}$.
    Substituting this into~\eqref{eq:R_star_proof_1} gives $I(X;T_1) = R^*$.
\end{proof}

We can now define the budget $R$ and the threshold $\vartheta$. Since
$R^*$ in~\eqref{eq:R_star} is in general irrational, we take a rational number $\tilde R$ such
that $|\tilde R - R^*| \le \frac{\delta}{4}$, and we set
\begin{equation}\label{eq:R_def}
    R = \max\Big\{ \tilde R + \frac{\delta}{4},\ 0 \Big\}.
\end{equation}
The additive term $\frac{\delta}{4}$ ensures that $R \ge R^*$, so that
the encoder $T_1$ satisfies the rate constraint whenever $G$ contains a $k$-clique,
while the maximum ensures that $R \ge 0$. To choose the threshold
$\vartheta$, we compare the two cases. If $\omega(G) \ge k$, then, by
Lemma~\ref{lemma:R_star}, $R^* = I(X;T_1) \ge 0$, so that the maximum
in~\eqref{eq:R_def} equals $\tilde R + \frac{\delta}{4}$, and, since
$|\tilde R - R^*| \le \frac{\delta}{4}$, we have that
\begin{equation}\label{eq:R_window}
    R^* \le R \le R^* + \frac{\delta}{2}.
\end{equation}
Hence, by Lemma~\ref{lemma:IB_yes} and $R^*\geq R-\delta/2$ from the rightmost inequality
in~\eqref{eq:R_window}, we obtain that
\begin{equation}\label{eq:IBT_T1_value}
    I(Y;T_1) > \frac{R^*}{2} + \Pi_0 + \delta
    \ge \frac12\Big( R - \frac{\delta}{2} \Big) + \Pi_0 + \delta
    = \frac{R}{2} + \Pi_0 + \frac{3\delta}{4}.
\end{equation}
On the other hand, if $\omega(G) < k$, then
$\IB(P_{XY},R) \le \frac{R}{2} + \Pi_0$ by Lemma~\ref{lemma:IB_no}.
Thus, as in the proof of Theorem~\ref{thm:IBL_hard}, we place the threshold in
the middle of this gap, and, since this value is in general irrational,
we take $\vartheta$ as a rational number such that
\begin{equation}\label{eq:theta_def}
    \Big| \vartheta - \Big( \frac{R}{2} + \Pi_0 + \frac{3\delta}{8} \Big) \Big|
    \le \frac{\delta}{8}.
\end{equation}
Therefore, the instance of \IBT{} associated with $(G,k)$ is $(P_{XY},R,\vartheta)$, where $P_{XY}$ is defined in~\eqref{eq:IB_joint}, $R$ in~\eqref{eq:R_def} and $\vartheta$ in~\eqref{eq:theta_def}.

\begin{proof}[Proof of Theorem~\ref{thm:IBT_hard}]
     We start by showing that \IBT{} is NP-hard. To this end, we show that the construction of $(P_{XY},R,\vartheta)$ is a
    polynomial-time reduction from \clique{} to \IBT{}.

    We first show that the reduction runs in polynomial time. As shown in
    the proof of Theorem~\ref{thm:IBL_hard}, $\G$ and $P_{XY}$ can be
    constructed in polynomial time, so it remains to show that $R$ and
    $\vartheta$ can be computed in polynomial time. By~\eqref{eq:R_star}, $R^*$ is a linear combination, with
    rational coefficients, of the logarithms of $N$, $m+1$, $m-1$, $s$,
    $1/a_0$ and $1/a_1$, which are rational numbers with $O(\log n)$ bits.
    Hence, by Lemma~\ref{lem:log} in Appendix~\ref{app:log}, $\tilde R$ can be computed in
    polynomial time, and so can $R$ by~\eqref{eq:R_def}. Similarly,
    $\frac{R}{2} + \Pi_0 + \frac{3\delta}{8}$ is the sum of a rational
    number and of $\Pi_0$, which, as shown in the proof of
    Theorem~\ref{thm:IBL_hard}, is an explicit expression involving
    polynomially many logarithms. Hence, again by Lemma~\ref{lem:log}, a
    rational $\vartheta$ satisfying~\eqref{eq:theta_def} can be computed
    in polynomial time. Therefore, the whole reduction runs in polynomial
    time.

    We now show that the reduction is correct, that is,
    $(P_{XY},R,\vartheta)$ is a yes-instance of \IBT{} if and only if $(G,k)$ is a yes-instance of \clique{}.  Recall that
    $(P_{XY}, R, \vartheta)$ is a yes-instance of \IBT{} if there exists an
    encoder $T$ of $X$ such that $I(X;T) \le R$ and
    $I(Y;T) \ge \vartheta$, and that $(G,k)$ is a yes-instance of \clique{} if $\omega(G)\geq k$. We notice that, by~\eqref{eq:theta_def}, the threshold $\vartheta$ satisfies
    \begin{equation}\label{eq:theta_window}
        \frac{R}{2} + \Pi_0 + \frac{\delta}{4} \le \vartheta
        \le \frac{R}{2} + \Pi_0 + \frac{\delta}{2}.
    \end{equation}
    Assume first that $(G,k)$ is a yes-instance, that is, $\omega(G) \ge k$. We show that the encoder $T_1$ satisfies both $I(X;T_1)\leq R$ and $I(Y;T_1)\geq \vartheta$. By Lemma~\ref{lemma:R_star} and the leftmost
    inequality in~\eqref{eq:R_window}, we have that $I(X;T_1) = R^* \le R$, and therefore the encoder $T_1$
    satisfies the rate constraint.
    Moreover, by~\eqref{eq:IBT_T1_value} and the rightmost inequality
    in~\eqref{eq:theta_window}, we get that
    $I(Y;T_1) > \frac{R}{2} + \Pi_0 + \frac{3\delta}{4} \ge \vartheta + \frac{\delta}{4}$.
   Hence, $T_1$ satisfies both $I(X;T_1) \le R$ and
    $I(Y;T_1) \ge \vartheta$, and $(P_{XY}, R, \vartheta)$ is a yes-instance
    of \IBT{}. Moreover,
    $\IB(P_{XY},R) \geq I(Y;T_1)>\vartheta + \frac{\delta}{4}$.

    Assume now that $(G,k)$ is a no-instance, that is, $\omega(G) < k$. By Lemma~\ref{lemma:IB_no} and the
    leftmost inequality in~\eqref{eq:theta_window}, we get that
    $\IB(P_{XY},R) \le \frac{R}{2} + \Pi_0 \le \vartheta - \frac{\delta}{4}$.
    Hence, every encoder with $I(X;T) \le R$ satisfies
    $I(Y;T) < \vartheta$, and $(P_{XY},R,\vartheta)$ is a no-instance of \IBT{}.
    This proves that the reduction is correct, and hence that \IBT{} is
    NP-hard.

    It remains to prove the inapproximability result. By the two cases
    above, $\IB(P_{XY},R)$ lies above $\vartheta + \frac{\delta}{4}$ if
    $\omega(G) \ge k$, and below $\vartheta - \frac{\delta}{4}$ otherwise.
    Suppose that a polynomial-time algorithm outputs $v$ with
    $|v - \IB(P_{XY},R)| \le \frac{1}{144M^{3/2}}$. By
    Lemma~\ref{lemma:delta_bound}, this error is at most
    $\frac{\delta}{8} < \frac{\delta}{4}$. Hence, if $\omega(G) \ge k$,
    then $v \ge \IB(P_{XY},R) - \frac{\delta}{8} > \vartheta$, while, if
    $\omega(G) < k$, then $v \le \IB(P_{XY},R) + \frac{\delta}{8} < \vartheta$.
    Thus,  running the
    reduction, then the algorithm, and comparing $v$ with $\vartheta$ would decide \clique{} in
    polynomial time, which is impossible unless
    $\mathrm{P} = \mathrm{NP}$.
\end{proof}

\subsection{A better exponent}\label{subsec:padding}

In this section, we show that the exponent $3/2$ in Theorems~\ref{thm:IBL_hard}
and~\ref{thm:IBT_hard} and Proposition~\ref{prop:other_beta} can be lowered to $1 + \alpha$, for every fixed
$\alpha > 0$. Recall that a smaller exponent means a larger additive error,
and hence a stronger inapproximability result.

To explain where the exponent $3/2$ comes from, recall
from~\eqref{eq:delta_M} that the gap of the reduction is
$\delta = \frac{1 - 1/m}{8Mm^2}$. The factor $\frac1M$ is of the desired
order, while the factor $\frac{1}{m^2}$ has to be bounded in terms of $M$.
In the proof of Lemma~\ref{lemma:delta_bound}, the only relation between
$m$ and $M$ comes from the assumption $N \ge B \approx m^2/2$
in~\eqref{eq:assumption_on_G}, which gives $m^2 \le \frac32 \sqrt{M}$,
and hence $\delta \ge \frac{1}{18 M^{3/2}}$. The idea to obtain a better result is therefore to replace the instance $(G,k)$ with an equivalent
instance $(G',k)$, that is, one with the same answer to \clique{}, in
which $G'$ has many more edges than $G$. In this way, $m^2$ becomes small
with respect to $M=|A(\G')|$. This is formalized in the following theorem.

\begin{theorem}\label{thm:padding}
    Let $\alpha > 0$ be fixed. Unless $\mathrm{P} = \mathrm{NP}$, there is
    no polynomial-time algorithm that, given a rational joint distribution
    $P_{XY}$ on $\cX \times \cY$, approximates $L_{1/2}(P_{XY})$ within an
    additive error of $\frac{1}{36|\cX|^{1+\alpha}}$, nor one that, given
    also a non-negative rational number $R$, approximates
    $\IB(P_{XY},R)$ within $\frac{1}{96|\cX|^{1+\alpha}}$. Moreover, for
    every fixed $\beta \in \mathbb{Q} \cap (0,1)$, there is no
    polynomial-time algorithm that approximates $L_\beta(P_{XY})$ within
    $\frac{\kappa_\beta}{36|\cX|^{1+\alpha}}$.
\end{theorem}

\begin{proof}
    We first construct the new instance of \clique{}. Let
    $c = \lceil 1/\alpha \rceil$, so that $c$ is an integer and
    $1/c \le \alpha$, and let $(G,k)$ be an instance of \clique{}
    satisfying~\eqref{eq:assumption_on_G}, with $n = |V(G)|$. Let
    $N' = n^{2c}$, and let $u_1, v_1, \dots, u_{N'}, v_{N'}$ be $2N'$ new
    vertices, not in $V(G)$. We define the graph $G'$ by
    \begin{equation*}
        V(G') = V(G) \cup \{u_1, v_1, \dots, u_{N'}, v_{N'}\}
        \quad\text{and}\quad
        E(G') = E(G) \cup \big\{ \{u_i, v_i\} : 1 \le i \le N' \big\},
    \end{equation*}
    that is, $G'$ is the disjoint union of $G$ and of $N'$ isolated edges.
    We claim that $G'$ contains a $k$-clique if and only if $G$ does.
    Indeed, every clique of $G$ is a clique of $G'$. Conversely, let $C$
    be a clique of $G'$ that contains a new vertex, say $u_i$. Since the
    only neighbour of $u_i$ in $G'$ is $v_i$, we have that
    $C \subseteq \{u_i, v_i\}$, and hence $|C| \le 2 < k$. The same holds
    for $v_i$. Therefore, every $k$-clique of $G'$ is contained in $V(G)$,
    and it is a $k$-clique of $G$. Moreover, $(G',k)$
    satisfies~\eqref{eq:assumption_on_G}, since
    $|E(G')| = N + N' \ge N \ge B$. Finally, $G'$ has $n + 2n^{2c}$
    vertices, and, since $c$ is fixed, its adjacency matrix can be
    constructed in time polynomial in $n$.

    We now apply the reductions of Theorems~\ref{thm:IBL_hard}
    and~\ref{thm:IBT_hard} to $(G',k)$ instead of $(G,k)$. All the results
    of Sections~\ref{sec:construction} and~\ref{sec:IB} hold for every
    instance satisfying~\eqref{eq:assumption_on_G}, and hence for $(G',k)$ as well.
    Moreover, since the size of $G'$ is polynomial in $n$, both reductions
    still run in time polynomial in $n$. The only change concerns the lower
    bound on $\delta$. Since $m = k - 1 < n$, the extended graph of $G'$
    has
    \begin{equation*}
        M = |E(G')|\, m(m-1) \ge N' = n^{2c} > m^{2c}
    \end{equation*}
    oriented edges, that is, $m^2 \le M^{1/c}$. Using this bound
    in~\eqref{eq:delta_M} in place of~\eqref{eq:m_upper}, together with
    $1 - \frac1m \ge \frac23$ and $1/c \le \alpha$, we get that
    \begin{equation}\label{eq:delta_padding}
        \delta \ge \frac{2/3}{8\, M \cdot M^{1/c}}
        = \frac{1}{12\, M^{1+1/c}} \ge \frac{1}{12\, M^{1+\alpha}}.
    \end{equation}
    The proof of Theorem~\ref{thm:IBL_hard} rules out polynomial-time
    approximation of $L_{1/2}(P_{XY})$ within $\frac{\delta}{3}$, and the
    proof of Theorem~\ref{thm:IBT_hard} rules out polynomial-time
    approximation of $\IB(P_{XY},R)$ within $\frac{\delta}{8}$. Since
    $M = |\cX|$, by~\eqref{eq:delta_padding} we have that
    $\frac{\delta}{3} \ge \frac{1}{36|\cX|^{1+\alpha}}$ and
    $\frac{\delta}{8} \ge \frac{1}{96|\cX|^{1+\alpha}}$, which proves the
    first two claims. The last claim follows from the first one exactly as
    in the proof of Proposition~\ref{prop:other_beta}, since the
    construction~\eqref{eq:Z_def} does not change the alphabet $\cX$ and
    multiplies the Lagrangian by $\kappa_\beta$.
\end{proof}

Note that Theorem~\ref{thm:padding} improves on
Theorems~\ref{thm:IBL_hard} and~\ref{thm:IBT_hard} only for
$0<\alpha < 1/2$, that is, for $c \ge 3$, and that it becomes stronger as
$\alpha$ decreases, at the price of a larger graph $G'$. The exponent cannot
be lowered to $1$ within this construction. Indeed, the improvement $\delta$ of $T_1$ over $T_0$ is of order
$\frac{1}{Nm^4} = \Theta\big(\frac{1}{Mm^2}\big)$, and $m$ cannot be
bounded, since \clique{} can be solved in polynomial time for every
fixed $k$. 

\section{Hardness of the privacy funnel}\label{sec:PF}

In this section, we prove our results for the privacy funnel. For the
constrained formulation~\eqref{eq:PF_def}, we prove the following theorem.

\begin{theorem}\label{thm:PF_hard}
    \PFT{} is NP-hard. Moreover, unless $\mathrm{P} = \mathrm{NP}$, there
    is no polynomial-time algorithm that, given a rational joint
    distribution $P_{XY}$ on $\cX \times \cY$ and a rational
    number $0\leq r\leq H(X)$, outputs a value $v$ such that
    \begin{equation}\label{eq:PF_approx}
        \big| v - \PF(P_{XY},r) \big| \le \frac{1}{4|\cX|^2}.
    \end{equation}
\end{theorem}

Theorem~\ref{thm:PF_hard} is proved by a reduction from \clique{}, which
reuses the extended graph $\G$, the random variable $X$, and the encoders
$T_0$ and $T_1$ of Section~\ref{sec:construction}. Only the channel from
$X$ to $Y$, and hence the joint distribution $P_{XY}$, changes. A key
tool for the reduction is the lower bound of Lemma~\ref{lemma:PF_lower}, namely
$I(Y;T) \ge I(X;T) - H(X \mid Y)$, which holds for every encoder $T$ of
$X$.

The rest of this section is organized as follows. In
Section~\ref{subsec:compatible}, we introduce compatible sets of
oriented edges, that is, sets of oriented edges of a common clique, and
we define a function $\Phi$ of the distributions on the oriented edges,
which plays a similar role to that of the cost $g_W$ of Section~\ref{sec:IB}. We then
bound $\Phi$ in terms of the clique number of $\G$. In
Section~\ref{subsec:PF_channel}, we define the joint distribution of the
reduction, and we show that an encoder attains the lower bound
of Lemma~\ref{lemma:PF_lower} if its posteriors are supported on
compatible sets. In Section~\ref{subsec:PF_reduction}, we show that, if
$G$ has no $k$-clique, the lower bound of Lemma~\ref{lemma:PF_lower} cannot be attained at any rate
smaller than $\log N$, while, if $G$ contains a $k$-clique, the encoder
$T_1$ attains the lower bound at a strictly smaller rate. Finally, we use these
results to prove Theorem~\ref{thm:PF_hard} in
Section~\ref{subsec:PF_proof}, and in
Section~\ref{subsec:PF_lagrangian}, we extend our results to the
Lagrangian formulation~\eqref{eq:Lambda_def}
(Proposition~\ref{prop:PF_lagrangian}). As for the information
bottleneck, we made no attempt to optimize the constants in these
results.

\subsection{Compatible sets and an entropic bound}\label{subsec:compatible}

Let $(G,k)$ be an instance of \clique{} satisfying~\eqref{eq:assumption_on_G},
and let $\G$ be the extended graph of Section~\ref{subsec:extended_graph}.
We say that two distinct oriented edges $x, x' \in A(\G)$ are
\emph{compatible} if the set of their endpoints is a clique of $\G$, that
is, if $x$ and $x'$ can be oriented edges of the same clique; otherwise,
they are \emph{incompatible}. A set $S \subseteq A(\G)$ is
\emph{compatible} if any two of its elements are compatible. We denote by
$\cI(\G)$ the set of unordered pairs $\{x,x'\}$ of incompatible oriented
edges of $\G$.

The next lemma shows that a set of oriented edges is compatible if and
only if it is contained in a clique of $\G$. Consequently, the size of
a compatible set is bounded in terms of the clique number of $\G$.
\begin{lemma}\label{lemma:PF_compatible_sets}
    A set $S \subseteq A(\G)$ is compatible if and only if
    $S \subseteq A(C)$ for some clique $C$ of $\G$. Moreover, every
    compatible set $S$ satisfies $|S| \le \omega(\G)(\omega(\G)-1)$.
\end{lemma}

\begin{proof}
    We first prove that, if $S \subseteq A(C)$ for some clique $C$ of
    $\G$, then $S$ is compatible. Let $x, x' \in S$ be distinct. Since
    $x, x' \in A(C)$, all their endpoints belong to $C$. Moreover, since
    $C$ is a clique and every subset of a clique is still a clique, it
    follows that the set of the endpoints of $x$ and $x'$ is a clique of
    $\G$, which means that $x$ and $x'$ are compatible. Hence, $S$ is
    compatible.

    Conversely, assume that $S$ is compatible. We show that
    $S \subseteq A(C)$ for some clique $C$ of $\G$. Let
    \begin{equation}\label{eq:C_from_S}
        C = \{ u \in V(\G) : u \text{ is an endpoint of some } x \in S \}.
    \end{equation}
    By definition of $C$, we have that $S \subseteq A(C)$, so it suffices
    to show that $C$ is a clique of $\G$. Let $u, v \in C$ be distinct, and
    let $x, x' \in S$ be oriented edges such that $u$ is an endpoint of $x$
    and $v$ is an endpoint of $x'$. If $x = x'$, then $u$ and $v$ are the
    two endpoints of $x$, that is, $x = (u,v)$ or $x = (v,u)$. Since
    $x \in A(\G)$, we have that $\{u,v\} \in E(\G)$, so $u$ and $v$ are
    adjacent. If $x \ne x'$, then $x$ and $x'$ are compatible, since $S$ is
    compatible, and the set of their endpoints, which contains $u$ and
    $v$, is a clique of $\G$. Hence, $u$ and $v$ are adjacent also in this
    case. This implies that $C$ is a clique of $\G$. Hence, we proved that a set $S \subseteq A(\G)$ is compatible if and only if
    $S \subseteq A(C)$ for some clique $C$ of $\G$.

    Moreover, if $S$ is compatible, then $S \subseteq A(C)$ for some clique
    $C$ of $\G$, and hence it follows that
    \begin{equation*}
        |S| \le |A(C)| = |C|(|C|-1) \le \omega(\G)(\omega(\G)-1)
    \end{equation*}
    by definition of $\omega$ in~\eqref{eq:clique_number_def}.
\end{proof}

For a probability distribution $Q$ on $A(\G)$, let
\begin{equation}\label{eq:Phi_def}
    \Phi(Q) = H(Q) - \sum_{\{x,x'\} \in \cI(\G)}
    \big( Q(x) + Q(x') \big)\, h_2\Big( \frac{Q(x)}{Q(x) + Q(x')} \Big),
\end{equation}
where we assume that each term of the sum with $Q(x) + Q(x') = 0$ is
equal to zero. In words, $\Phi(Q)$ is the entropy of $Q$ minus a penalty
for each pair of incompatible oriented edges that both have positive
probability under $Q$. In particular, if $\supp(Q)$ is compatible, then
all the penalties vanish and $\Phi(Q) = H(Q)$. In the next lemma, we
bound $\Phi$ in terms of the clique number of $\G$. This bound plays, for
the privacy funnel, the same role that the bound on $g_{\G}$ in Lemma~\ref{lemma:entropic_turan}
plays for the information bottleneck.

\begin{lemma}\label{lemma:Phi_max}
    For every probability distribution $Q$ on $A(\G)$, it holds that
    \begin{equation}\label{eq:Phi_max}
        \Phi(Q) \le \log\big( \omega(\G)(\omega(\G)-1) \big).
    \end{equation}
\end{lemma}

\begin{proof}
    Let $Q$ be a probability distribution on $A(\G)$. We show that,
    whenever $\supp(Q)$ is not compatible, that is, it contains both
    elements of at least one pair $\{x,x'\}$ in $\cI(\G)$, we can move all
    the probability of $x$ to $x'$, or of $x'$ to $x$, without decreasing
    the value of $\Phi$. Reiterating this operation, we reach a probability
    distribution $Q'$ with compatible support and $\Phi(Q') \ge \Phi(Q)$.
    Then, the claim follows from Lemma~\ref{lemma:PF_compatible_sets}.

    Assume that $Q(x) > 0$ and $Q(x') > 0$ for some
    $\{x,x'\} \in \cI(\G)$, and let $c = Q(x) + Q(x')$. For $a \in [0,c]$,
    let $Q_a$ be the probability distribution that coincides with $Q$
    except that $Q_a(x) = a$ and $Q_a(x') = c - a$. Note that
    $Q = Q_{Q(x)}$, and that $Q_0$ and $Q_c$ are the probability
    distributions obtained from $Q$ by moving all the probability of $x$
    to $x'$, and of $x'$ to $x$, respectively. We show that the function
    $a \mapsto \Phi(Q_a)$ is convex on $[0,c]$. To this end, we recall
    that
    \begin{equation}\label{eq:phi_Q_a_recall}
        \Phi(Q_a) = H(Q_a) - \sum_{\{y,y'\} \in \cI(\G)}
        \big( Q_a(y) + Q_a(y') \big)\, h_2\Big( \frac{Q_a(y)}{Q_a(y) + Q_a(y')} \Big).
    \end{equation}
    We now split the terms of~\eqref{eq:phi_Q_a_recall} that depend on
    $a$ into three groups.

    First, since $Q_a$ coincides with $Q$ outside $\{x,x'\}$, the terms of
    $H(Q_a)$ that depend on $a$ are
    \begin{equation}\label{eq:first_term_H_Q_a}
        -a\log a - (c-a)\log(c-a).
    \end{equation}
    Second, the term of the sum in~\eqref{eq:phi_Q_a_recall} corresponding
    to the pair $\{y,y'\} = \{x,x'\}$ is equal to
    \begin{equation}\label{eq:second_term_sum}
        c\, h_2\Big(\frac{a}{c}\Big) = -a\log a - (c-a)\log(c-a) + c\log c.
    \end{equation}
    Hence, the difference between~\eqref{eq:first_term_H_Q_a}
    and~\eqref{eq:second_term_sum} gives the term $-c\log c$, which is
    constant and does not depend on $a$.

    Third, the only other terms of the sum in~\eqref{eq:phi_Q_a_recall}
    that depend on $a$ are those corresponding to pairs $\{x,z\}$ or
    $\{x',z\}$ in $\cI(\G)$, with $z \notin \{x,x'\}$. Each of these terms
    has the form
    \begin{equation*}
        f(b) = (b+w)\, h_2\Big(\frac{b}{b+w}\Big)
        = -b\log b - w\log w + (b+w)\log(b+w),
    \end{equation*}
    where $w = Q(z)$, which does not depend on $a$, and $b = a$ or
    $b = c - a$. Since
    $f''(b) = \frac{1}{\ln 2}\big(\frac{1}{b+w} - \frac1b\big) \le 0$ for
    $b > 0$, the function $f$ is concave, and consequently each of these
    terms is a concave function of $a$.

    Putting the three groups together, we get that, as a function of
    $a \in [0,c]$, $\Phi(Q_a)$ is equal to a constant minus a sum of
    concave functions of $a$. Hence, $a \mapsto \Phi(Q_a)$ is convex on
    $[0,c]$, and, since a convex function on an interval attains its
    maximum at an endpoint, we get that
    \begin{equation}\label{eq:Phi_endpoints}
        \Phi(Q) = \Phi(Q_{Q(x)}) \le \max\{ \Phi(Q_0), \Phi(Q_c) \}.
    \end{equation}
    Since both $Q_0$ and $Q_c$ have one element less in their support than
    $Q$, by repeating this operation, after at most $|\supp(Q)| - 1$ steps
    we obtain a probability distribution $Q'$ with $\Phi(Q') \ge \Phi(Q)$
    whose support is compatible. For such a distribution $Q'$, all the
    terms of the sum in~\eqref{eq:Phi_def} vanish, so that
    $\Phi(Q') = H(Q')$. Hence, since $\supp(Q')$ is a compatible set,  by Lemma~\ref{lemma:PF_compatible_sets} we
    finally obtain that
    \begin{equation*}
        \Phi(Q) \le \Phi(Q') = H(Q') \le \log|\supp(Q')|
        \le \log\big( \omega(\G)(\omega(\G)-1) \big),
    \end{equation*}
    which concludes the proof.
\end{proof}

\subsection{The joint distribution}\label{subsec:PF_channel}

We now describe the joint distribution of the instance of \PFT{}. As in
Section~\ref{sec:construction}, we let the random variable $X$ be uniformly distributed on
$\cX = A(\G)$, with $|A(\G)| = M = Ns$, where $N=|E(G)|$ and $s=m(m-1)$. Moreover, we let
\begin{equation}\label{eq:R_0}
    R_0 = \log N,
\end{equation}
which, by Lemma~\ref{lemma:T_0}, is the rate of the encoder $T_0$, that is, $I(X;T_0)=R_0$.

We now define the random variable $Y$. For each pair
$\{x,x'\} \in \cI(\G)$, we introduce a symbol $y_{\{x,x'\}}$. For each
$x \in A(\G)$, let $d(x)$ be the number of oriented edges incompatible
with $x$, and note that $d(x) \le M - 1$. We also introduce $M - d(x)$
symbols $(x,1), \dots, (x, M - d(x))$. The alphabet $\cY$ of $Y$ is the
set of all these symbols. For every $x \in A(\G)$, let
\begin{equation}\label{eq:PF_Y_x}
    \cY_x = \big\{ y_{\{x,z\}} : z \in A(\G),\ \{x,z\} \in \cI(\G) \big\}
    \cup \big\{ (x,1), \dots, (x, M - d(x)) \big\}
\end{equation}
be the set of values that $Y$ can take when $X = x$, and note that
$|\cY_x| = d(x) + (M - d(x)) = M$. Given $X = x$, we let $Y$ be uniformly
distributed on $\cY_x$, that is,
\begin{equation}\label{eq:PF_channel}
    W(y \mid x) = P_{Y\mid X}(y \mid x) =
    \begin{cases}
        1/M, & \text{if } y \in \cY_x,\\
        0, & \text{otherwise}.
    \end{cases}
\end{equation}
Consequently, since $P_X$ is uniform on $A(\G)$, the joint distribution $P_{XY}$
of $(X,Y)$ is
\begin{equation}\label{eq:PF_joint}
    P_{XY}(x,y) = P_X(x)\, W(y \mid x) =
    \begin{cases}
        \dfrac{1}{M^2}, & \text{if } y \in \cY_x,\\[6pt]
        0, & \text{otherwise},
    \end{cases}
    \qquad \forall (x,y) \in A(\G) \times \cY.
\end{equation}
By~\eqref{eq:PF_Y_x} and~\eqref{eq:PF_joint}, we notice that, after observing $Y$, we
can draw the following conclusions about $X$. If $Y = (x,i)$, then, since
the symbol $(x,i)$ belongs to $\cY_x$ only, $X = x$ with certainty. If
$Y = y_{\{x,x'\}}$, then, since the symbol $y_{\{x,x'\}}$ belongs to both
$\cY_x$ and $\cY_{x'}$, and to no other set $\cY_z$, $X$ is equal to $x$
or to $x'$, each with probability $1/2$. Consequently, $Y$ determines
$X$, unless $Y$ is a symbol $y_{\{x,x'\}}$, in which case $X$ is one of
the two incompatible oriented edges $x$ and $x'$. This property is
crucial in our reduction. Indeed, if we also know that $X$ belongs to a
compatible set $S$, then at most one of $x$ and $x'$ belongs to $S$, and
hence $X$ is determined by $Y$. Let
\begin{equation}\label{eq:PF_h0}
    h_0 = H(X \mid Y)
\end{equation}
be the uncertainty about $X$ that remains after observing $Y$. The next lemma provides some useful properties of the joint distribution $P_{XY}$ in~\eqref{eq:PF_joint} and of the channel $W$ in~\eqref{eq:PF_channel}.

\begin{lemma}\label{lemma:PF_channel}
    Let $P_{XY}$ be the joint distribution defined in~\eqref{eq:PF_joint},
    and let $W$ be the channel from $\cX$ to $\cY$ defined in~\eqref{eq:PF_channel}. Then, the
    following properties hold:
    \begin{enumerate}
        \item $h_0 = 2|\cI(\G)|/M^2$ and $H(Y) = 2\log M - h_0$;
        \item for every probability distribution $Q$ on $A(\G)$,
        \begin{equation}\label{eq:PF_entropy}
            H(WQ) = \log M + \Big(1 - \frac1M\Big) H(Q) + \frac{\Phi(Q)}{M};
        \end{equation}
        \item if every posterior of an encoder $T$ of $X$ is supported on a
        compatible set, then $I(Y;T) = I(X;T) - h_0$.
    \end{enumerate}
\end{lemma}

\begin{proof}
    We prove the three properties separately.

    \emph{Property 1.} We first compute $h_0$. By definition of the
    conditional entropy in~\eqref{eq:conditional_entropy}, we have that
    \begin{equation}\label{eq:h0_def_sum}
        h_0 = H(X \mid Y) = \sum_{y \in \cY} P_Y(y)\, H(X \mid Y = y).
    \end{equation}
    As observed before, if $y = (x,i)$, then, by~\eqref{eq:PF_joint},
    $P_{XY}(z,y) > 0$ only for $z = x$. Therefore, $X$ is determined by
    $Y$, and $H(X \mid Y = y) = 0$. If, instead, $y = y_{\{x,x'\}}$, then,
    by~\eqref{eq:PF_joint},
    $P_{XY}(x, y_{\{x,x'\}}) = P_{XY}(x', y_{\{x,x'\}}) = 1/M^2$ and
    $P_{XY}(z, y_{\{x,x'\}}) = 0$ for every other $z$. Hence,
    $P_Y(y_{\{x,x'\}}) = 2/M^2$, and, given $Y = y_{\{x,x'\}}$, $X$ is
    uniformly distributed on $\{x,x'\}$, so that
    $H(X \mid Y = y_{\{x,x'\}}) = 1$. Consequently, since there are
    $|\cI(\G)|$ symbols of the form $y_{\{x,x'\}}$, splitting the sum
    in~\eqref{eq:h0_def_sum} according to the two kinds of symbols in $\cY$, we
    obtain that
    \begin{align*}
        h_0 &= \sum_{x \in A(\G)} \sum_{i=1}^{M - d(x)} P_Y\big((x,i)\big)\,
               H\big(X \mid Y = (x,i)\big)
             + \sum_{\{x,x'\} \in \cI(\G)} P_Y\big(y_{\{x,x'\}}\big)\,
               H\big(X \mid Y = y_{\{x,x'\}}\big)\\
            &= \sum_{x \in A(\G)} \sum_{i=1}^{M - d(x)} P_Y\big((x,i)\big) \cdot 0
             + \sum_{\{x,x'\} \in \cI(\G)} \frac{2}{M^2} \cdot 1
             = |\cI(\G)| \cdot \frac{2}{M^2} = \frac{2|\cI(\G)|}{M^2}.
    \end{align*}

    We now compute $H(Y)$. By the chain rule~\eqref{eq:chain_rule},
    $H(X,Y) = H(X) + H(Y \mid X) = H(Y) + H(X \mid Y)$, and hence
    \begin{equation}\label{eq:H_Y_proof}
        H(Y) = H(X) + H(Y \mid X) - H(X \mid Y).
    \end{equation}
    Since $X$ is uniform on $A(\G)$, we have that $H(X) = \log M$, and,
    by~\eqref{eq:PF_channel}, given $X = x$, $Y$ is uniform on the $M$
    elements of $\cY_x$, so that $H(Y \mid X) = \log M$. Putting these results together and recalling that
    $H(X \mid Y) = h_0$, from~\eqref{eq:H_Y_proof} we obtain that $H(Y) = 2\log M - h_0$.

    \emph{Property 2.} Let $Q$ be a probability distribution on $A(\G)$.  We compute $H(WQ)$ explicitly. Let $a_x = Q(x)/M$
    for every $x \in A(\G)$. By the definition of the channel $W$ in~\eqref{eq:PF_channel}, each symbol
    $(x,i)$ has probability $(WQ)\big((x,i)\big) = a_x$, and each symbol
    $y_{\{x,x'\}}$ has probability $(WQ)\big(y_{\{x,x'\}}\big) = a_x + a_{x'}$.
    Since there are $M - d(x)$ symbols of the form $(x,i)$ for each
    $x \in A(\G)$, we obtain that
    \begin{equation}\label{eq:PF_HWQ_1}
        H(WQ) = -\sum_{x \in A(\G)} \big(M - d(x)\big)\, a_x \log a_x
        - \sum_{\{x,x'\} \in \cI(\G)} (a_x + a_{x'}) \log(a_x + a_{x'}).
    \end{equation}
    To rewrite the first sum in~\eqref{eq:PF_HWQ_1}, we observe that the term $a_x \log a_x$ of a
    fixed $x\in A(\G)$ appears once in the sum
    $\sum_{\{x,x'\} \in \cI(\G)} \big( a_x \log a_x + a_{x'} \log a_{x'} \big)$
    for each pair of $\cI(\G)$ containing $x$, that is, exactly $d(x)$
    times. Hence, we have that
    \begin{equation}\label{eq:first_step_sum_x_x'}
        \sum_{x \in A(\G)} d(x)\, a_x \log a_x
        = \sum_{\{x,x'\} \in \cI(\G)} \big( a_x \log a_x + a_{x'} \log a_{x'} \big).
    \end{equation}
    Substituting~\eqref{eq:first_step_sum_x_x'} into~\eqref{eq:PF_HWQ_1},
    we get that
    \begin{equation}\label{eq:PF_HWQ_2}
        H(WQ) = -\sum_{x \in A(\G)} M\, a_x \log a_x
        - \sum_{\{x,x'\} \in \cI(\G)}
        \Big( (a_x + a_{x'}) \log(a_x + a_{x'}) - a_x \log a_x - a_{x'} \log a_{x'} \Big).
    \end{equation}
    We now evaluate the two sums in~\eqref{eq:PF_HWQ_2} separately. For
    the first one, since $M a_x = Q(x)$, we get that
    \begin{equation}\label{eq:PF_first_sum}
        -\sum_{x \in A(\G)} M\, a_x \log a_x
        = -\sum_{x \in A(\G)} Q(x) \log\frac{Q(x)}{M} = H(Q) + \log M.
    \end{equation}
    For the second one, we use the identity
    $(a+b)\log(a+b) - a\log a - b\log b = (a+b)\, h_2\big(\frac{a}{a+b}\big)$,
    valid for $a, b \ge 0$ with $a + b > 0$. Since
    $a_x + a_{x'} = \frac{Q(x) + Q(x')}{M}$ and
    $\frac{a_x}{a_x + a_{x'}} = \frac{Q(x)}{Q(x) + Q(x')}$, each term of the
    second sum equals
    $\frac{Q(x) + Q(x')}{M}\, h_2\big( \frac{Q(x)}{Q(x) + Q(x')} \big)$,
    and, by the definition of $\Phi$ in~\eqref{eq:Phi_def}, the whole second sum equals
    $\frac1M \big( H(Q) - \Phi(Q) \big)$. Substituting this
    and~\eqref{eq:PF_first_sum} into~\eqref{eq:PF_HWQ_2}, we obtain that
    \begin{equation*}
        H(WQ) = \log M + H(Q) - \frac1M \big( H(Q) - \Phi(Q) \big)
        = \log M + \Big(1 - \frac1M\Big) H(Q) + \frac{\Phi(Q)}{M},
    \end{equation*}
    which proves~\eqref{eq:PF_entropy}.

    \emph{Property 3.} Let $T$ be an encoder of $X$ whose posteriors are
    all supported on compatible sets. As in the proof of
    Lemma~\ref{lemma:PF_lower}, we have that
    \begin{equation}\label{eq:PF_prop3_identity}
        I(Y;T) = I(X;T) - I(X;T \mid Y) = I(X;T) - h_0 + H(X \mid Y,T).
    \end{equation}
    Hence, it suffices to show that $H(X \mid Y,T) = 0$, that is, that $X$
    is determined by $Y$ and $T$. By definition of conditional entropy (see \eqref{eq:conditional_entropy}), it
    reduces to showing that $H(X \mid Y, T = t) = 0$ for every $t$ with
    $P_T(t) > 0$. Fix such a $t$, and let $S=\supp(P_{X\mid T}(\cdot \mid t))$ be the support of $P_{X\mid T}(\cdot \mid t)$, which, by assumption, is compatible. Since $Y$ and
    $T$ are conditionally independent given $X$, we have that
    \begin{equation*}
        P_{XY\mid T}(x,y \mid t) = P_{X\mid T}(x \mid t)\, W(y \mid x),
    \end{equation*}
    where the factor $ P_{X\mid T}(x \mid t)$ is positive only if $x \in S$, since $S$ contains the
    support of $P_{X\mid T}(\cdot \mid t)$, and $ W(y \mid x)$ is positive only if $y \in \cY_x$, by~\eqref{eq:PF_channel}. Hence,
    $P_{XY\mid T}(x,y \mid t) > 0$ if and only if $x \in S$ and $y \in \cY_x$. We now show that,
    for every $y \in \cY$, there is at most one $x \in S$ with
    $y \in \cY_x$, which implies that, given $T = t$, the value of $Y$
    determines the value of $X$.  We distinguish two cases. If $y = (x,i)$ for some
    $x \in A(\G)$ and $1 \le i \le M - d(x)$, then $y$ belongs only to
    $\cY_x$, and consequently $x$ is the only oriented edge with $y \in \cY_x$. If
    $y = y_{\{x,x'\}}$ for some $\{x,x'\} \in \mathcal{I}(\G)$, then $y$
    belongs only to $\cY_x$ and $\cY_{x'}$. Since $x$ and $x'$ are
    incompatible, they cannot both belong to the compatible set $S$. Thus, at
    most one of them belongs to $S$. In both cases, there is at most one
    $x \in S$ with $y \in \cY_x$. Hence, given $T = t$ and $Y = y$, the
    random variable $X$ can take only one value, and therefore
    $H(X \mid Y, T = t) = 0$.
    Thus, by~\eqref{eq:PF_prop3_identity}, we conclude that
    $I(Y;T) = I(X;T) - h_0$.
\end{proof}

\subsection{The reduction}\label{subsec:PF_reduction}

We now show that the yes- and no-instances of \clique{} can be distinguished by comparing $\PF(P_{XY},r)$ with the lower bound of Lemma~\ref{lemma:PF_lower}. Indeed,  by Lemma~\ref{lemma:PF_lower}, every encoder $T$ of $X$ with $I(X;T) \ge r$ satisfies $I(Y;T) \ge I(X;T) - h_0 \ge r - h_0$, so that $\PF(P_{XY},r) \ge r - h_0$. We show that, for a rate $r$ slightly smaller than $R_0=\log N$, this lower bound is not attained if $G$ does not contain a $k$-clique, while it is attained if $G$ contains a $k$-clique.

\medskip
\subsubsection{No-instances: $G$ does not contain a $k$-clique}

Assume that $G$ has no $k$-clique. Then $\omega(\G) = m$ by property~3
of Lemma~\ref{lemma:gadget}, and hence, by Lemma~\ref{lemma:Phi_max},
$\Phi(Q) \le \log s =\log(m(m-1))$ for every probability distribution $Q$ on $A(\G)$.
Using this, the next lemma shows that no encoder $T$ of $X$ with
$I(X;T) < R_0$ attains the lower bound of Lemma~\ref{lemma:PF_lower}.
\begin{lemma}\label{lemma:PF_no}
    Assume that $\omega(G) < k$. Then every encoder $T$ of $X$ satisfies
    \begin{equation}\label{eq:PF_no}
        I(Y;T) \ge I(X;T) - h_0 + \frac{R_0 - I(X;T)}{M}.
    \end{equation}
\end{lemma}

\begin{proof}
    The idea of the proof is to bound the entropy of $Y$ given each value
    of $T$, using Lemma~\ref{lemma:Phi_max}, and then to average over $T$
    and express the result in terms of mutual information.

    We first bound the output entropy of the channel $W$ defined
    in~\eqref{eq:PF_channel}. Since $\omega(G) < k$, we have that
    $\omega(G) \le m$, and by property~3 of Lemma~\ref{lemma:gadget} it
    follows that $\omega(\G) = m$. Hence,
    $\omega(\G)(\omega(\G)-1) = m(m-1) = s$, and Lemma~\ref{lemma:Phi_max}
    gives $\Phi(Q) \le \log s$ for every probability distribution $Q$ on
    $A(\G)$. Substituting this bound into~\eqref{eq:PF_entropy} of
    Lemma~\ref{lemma:PF_channel}, we obtain that
    \begin{equation}\label{eq:PF_no_Q}
        H(WQ) \le \log M + \Big(1 - \frac1M\Big) H(Q) + \frac{\log s}{M}
    \end{equation}
    for every probability distribution $Q$ on $A(\G)$.

    We now apply~\eqref{eq:PF_no_Q} to the posterior distributions of an
    encoder. Let $T$ be an encoder of $X$, and let $t$ be such that
    $P_T(t) > 0$. By Lemma~\ref{lemma:posterior_Y_given_T}, we have that
    $P_{Y\mid T}(\cdot \mid t) = W P_{X\mid T}(\cdot \mid t)$. Hence,
    applying~\eqref{eq:PF_no_Q} with $Q = P_{X\mid T}(\cdot \mid t)$, we
    get that
    \begin{equation*}
        H\big(P_{Y\mid T}(\cdot \mid t)\big)
        \le \log M + \Big(1 - \frac1M\Big) H\big(P_{X\mid T}(\cdot \mid t)\big)
        + \frac{\log s}{M}.
    \end{equation*}
    Multiplying by $P_T(t)$, summing over $t$, by the definition 
    of conditional entropy, we
    obtain that
    \begin{equation}\label{eq:PF_no_cond}
        H(Y \mid T) \le \log M + \Big(1 - \frac1M\Big) H(X \mid T) + \frac{\log s}{M}.
    \end{equation}

    We now rewrite~\eqref{eq:PF_no_cond} in terms of mutual information.
    By property~1 of Lemma~\ref{lemma:PF_channel}, we have that
    $H(Y \mid T) = H(Y) - I(Y;T) = 2\log M - h_0 - I(Y;T)$. Moreover,
    since $X$ is uniform on $A(\G)$, we have that
    $H(X \mid T) = H(X) - I(X;T) = \log M - I(X;T)$. Substituting these
    two expressions into the inequality~\eqref{eq:PF_no_cond}, we get that
    \begin{equation*}
        2\log M - h_0 - I(Y;T)
        \le \log M + \Big(1 - \frac1M\Big)\big(\log M - I(X;T)\big) + \frac{\log s}{M}.
    \end{equation*}
    Solving for $I(Y;T)$ and expanding the product, we obtain that
    \begin{align}
        I(Y;T)
        &\ge \log M - h_0 - \Big(1 - \frac1M\Big)\log M
        + \Big(1 - \frac1M\Big) I(X;T) - \frac{\log s}{M} \nonumber\\
        &= I(X;T) - h_0 + \frac{\log M - \log s - I(X;T)}{M}. \label{eq:last_step_PF}
    \end{align}
    Since $\log M - \log s = \log(M/s) = \log N = R_0$
    by~\eqref{eq:R_0}, one has that the inequality~\eqref{eq:last_step_PF} proves~\eqref{eq:PF_no}, that is, 
    the lemma.
\end{proof}

Lemma~\ref{lemma:PF_no} states that, if $G$ has no $k$-clique, every
encoder $T$ exceeds the lower bound~\eqref{eq:PF_lower} of Lemma~\ref{lemma:PF_lower} by at least
$\big(R_0 - I(X;T)\big)/M$, which is positive whenever $I(X;T) < R_0$.
Hence, without a $k$-clique, no encoder with rate smaller than $R_0$
attains equality in~\eqref{eq:PF_lower}. On the other hand, the rate
$R_0$ itself is attained with equality by the encoder $T_0$. Indeed, by
Lemma~\ref{lemma:T_0}, $I(X;T_0) = R_0$, and the posteriors of $T_0$ are
the distributions $U_{C_e}$, which are supported on the sets $A(C_e)$.
Since each $C_e$ is a clique of $\G$, these sets are compatible by
Lemma~\ref{lemma:PF_compatible_sets}, and, by property~3 of
Lemma~\ref{lemma:PF_channel}, it holds that
\begin{equation}\label{eq:PF_T0}
    I(X;T_0) = R_0
    \qquad\text{and}\qquad
    I(Y;T_0) = R_0 - h_0.
\end{equation}

\medskip
\subsubsection{Yes-instances: $G$ contains a $k$-clique}

We now assume that $G$ contains a $k$-clique $K$, and we show that the
encoder $T_1$ of Section~\ref{subsec:T_1}, with a suitable choice of $p$,
attains equality in~\eqref{eq:PF_lower} at a rate strictly smaller than
$R_0$. To this end, we
choose
\begin{equation}\label{eq:PF_p}
    p = \frac{1}{2(m+1)},
    \qquad\text{so that}\qquad
    \varepsilon = 2\mu p = \frac{1}{s(m+1)}
    \quad\text{and}\quad
    \lambda = \frac{B\varepsilon}{N} = \frac{m}{2M},
\end{equation}
where the last equality follows from $B = \frac{m(m+1)}{2}$ and
$M = Ns$.

The next lemma shows that, with this choice of $p$, the rate
of the encoder $T_1$ is smaller than $R_0$ by at least $\frac{3}{4M}$.
\begin{lemma}\label{lemma:PF_yes}
    Assume that $\omega(G) \ge k$, and let $T_1$ be the encoder defined
    in~\eqref{eq:T_1_def} with $p$ as in~\eqref{eq:PF_p}. Then, $T_1$ satisfies
    \begin{equation}\label{eq:PF_yes}
        I(Y;T_1) = I(X;T_1) - h_0
        \qquad\text{and}\qquad
        I(X;T_1) < R_0 - \frac{3}{4M}.
    \end{equation}
\end{lemma}

\begin{proof}
    We first prove that the encoder $T_1$ attains equality in~\eqref{eq:PF_lower}, that is,  $I(Y;T_1) = I(X;T_1) - h_0$. By
    Lemma~\ref{lemma:T_1}, the posteriors of $T_1$ are $U_K$, which is
    supported on $A(K)$, and the distributions $V_e$ and $U_{C_e}$, which
    are supported on the sets $A(C_e)$ by~\eqref{eq:V_e_def}. Since $K$
    and the $C_e$ are cliques of $\G$, all these sets are compatible by
    Lemma~\ref{lemma:PF_compatible_sets}. Hence, since the posteriors of $T_1$ are supported on compatible sets, the equality
    $I(Y;T_1) = I(X;T_1) - h_0$ follows from property~3 of
    Lemma~\ref{lemma:PF_channel}.

    We now bound the rate of $T_1$. By Lemma~\ref{lemma:T_1_rate}
    and $R_0=\log N$, one has that
    \begin{equation}\label{eq:PF_I_X_T1}
        I(X;T_1) = R_0 - \eta,
        \qquad\text{where}\qquad
        \eta = \lambda \log\frac{m+1}{m-1} - \frac{B(1-\varepsilon)}{N}\, D_V
    \end{equation}
    is the increase in conditional entropy obtained by replacing $T_0$
    with $T_1$. Hence, to prove the rightmost inequality in~\eqref{eq:PF_yes} it suffices to show that $\eta > \frac{3}{4M}$. 
    To this end, we bound the two terms of $\eta$
    in~\eqref{eq:PF_I_X_T1} separately.

    We start with the first term. Using $\ln(1+z) \ge \frac{z}{1+z}$ with
    $z = \frac{2}{m-1}$, for which $1 + z = \frac{m+1}{m-1}$ and
    $\frac{z}{1+z} = \frac{2}{m+1}$, we get that
    $\ln\frac{m+1}{m-1} \ge \frac{2}{m+1}$. Since
    $\log z = \frac{\ln z}{\ln 2}$, this gives
    \begin{equation}\label{eq:PF_gain}
        \lambda \log\frac{m+1}{m-1} \ge \frac{2\lambda}{(m+1)\ln 2}.
    \end{equation}
    We now bound the second term of $\eta$ in~\eqref{eq:PF_I_X_T1}. 
    By~\eqref{eq:D_V_bound} in
    Lemma~\ref{lemma:perturbation}, rearranging the factors, we have that
    \begin{equation}\label{eq:PF_loss_1}
        \frac{B(1-\varepsilon)}{N}\, D_V
        \le \frac{B(1-\varepsilon)}{N} \cdot \frac{\varepsilon^2}{2\mu(1-\varepsilon)^2 \ln 2}
        = \frac{B\varepsilon}{N} \cdot \frac{\varepsilon}{\mu} \cdot \frac{1}{2(1-\varepsilon)\ln 2}.
    \end{equation}
    By the choice of $p$ in~\eqref{eq:PF_p}, we have that $\frac{B\varepsilon}{N} = \lambda$
    and $\frac{\varepsilon}{\mu} = \varepsilon s = \frac{1}{m+1}$.
    Moreover, since $m \ge 3$, we have that $s(m+1) \ge 24$, so that
    $\varepsilon = \frac{1}{s(m+1)} \le \frac{1}{24}$ and
    $\frac{1}{1-\varepsilon} \le \frac{24}{23}$. Substituting these facts
    into~\eqref{eq:PF_loss_1}, we obtain that
    \begin{equation}\label{eq:PF_loss}
        \frac{B(1-\varepsilon)}{N}\, D_V
        \le \frac{\lambda}{2(m+1)(1-\varepsilon)\ln 2}
        \le \frac{12}{23} \cdot \frac{\lambda}{(m+1)\ln 2}.
    \end{equation}
    We can now combine the two bounds. Substituting~\eqref{eq:PF_gain}
    and~\eqref{eq:PF_loss} into the expression of $\eta$
    in~\eqref{eq:PF_I_X_T1}, we get that
    \begin{equation}\label{eq:PF_eta_bound_1}
        \eta \ge \Big( 2 - \frac{12}{23} \Big) \frac{\lambda}{(m+1)\ln 2}
        = \frac{34}{23} \cdot \frac{\lambda}{(m+1)\ln 2}.
    \end{equation}
    Since $\lambda = \frac{m}{2M}$ by~\eqref{eq:PF_p} and
    $\frac{m}{m+1} \ge \frac34$ since $m \ge 3$, we have that
    $\frac{\lambda}{m+1} = \frac{m}{2(m+1)M} \ge \frac{3}{8M}$. Hence,
    using $\ln 2 < \frac{7}{10}$, it follows
    from~\eqref{eq:PF_eta_bound_1} that
    \begin{equation*}
        \eta > \frac{34}{23} \cdot \frac{3}{8M} \cdot \frac{10}{7}
        = \frac{255}{322\, M} > \frac{3}{4M},
    \end{equation*}
    where the last inequality follows from $4 \cdot 255 = 1020 > 966 = 3 \cdot 322$.
    Hence, by~\eqref{eq:PF_I_X_T1}, $I(X;T_1) = R_0 - \eta < R_0 - \frac{3}{4M}$, which concludes the proof.
\end{proof}

We can now define the instance of \PFT{} associated with $(G,k)$. We
choose the rate $r$ slightly smaller than $R_0$, namely such that
\begin{equation}\label{eq:PF_window}
    R_0 - \frac{3}{4M} \le r \le R_0 - \frac{5}{8M}.
\end{equation}
With this choice, the two cases of \clique{} lead to different values of
$\PF(P_{XY},r)$. Indeed, if $G$ has no $k$-clique, then, by
Lemma~\ref{lemma:PF_no} and the rightmost inequality in~\eqref{eq:PF_window},
every encoder $T$ of $X$ with $I(X;T) \ge r$ satisfies
$I(Y;T) \ge r - h_0 + \frac{R_0 - r}{M} \ge r - h_0 + \frac{5}{8M^2}$.
If, instead, $G$ contains a $k$-clique, then, by~\eqref{eq:PF_T0} and
Lemma~\ref{lemma:PF_yes}, the encoders $T_0$ and $T_1$ attain the lower
bound~\eqref{eq:PF_lower} of Lemma~\ref{lemma:PF_lower} at rates $R_0$ and less than
$R_0 - \frac{3}{4M}$, respectively. Since $r$ lies between these two
rates by~\eqref{eq:PF_window}, we show in the proof of
Theorem~\ref{thm:PF_hard} that, by combining $T_0$ and $T_1$, the lower
bound is also attained at rate $r$, so that $\PF(P_{XY},r) = r - h_0$.
Since the two cases are separated by a gap of $\frac{5}{8M^2}$, we place
the threshold in the middle of this gap, that is,
\begin{equation}\label{eq:PF_theta_def}
    \vartheta = r - h_0 + \frac{5}{16M^2}.
\end{equation}
Finally, since $R_0 = \log N$ is in general irrational, to obtain a
rational $r$ satisfying~\eqref{eq:PF_window}, we take a rational number
$\tilde R_0$ with $|\tilde R_0 - R_0| \le \frac{1}{16M}$, and we set
\begin{equation}\label{eq:PF_r_def}
    r = \tilde R_0 - \frac{11}{16M}.
\end{equation}
The instance of \PFT{} associated with $(G,k)$ is the triple
$(P_{XY}, r, \vartheta)$, where $P_{XY}$ is defined
in~\eqref{eq:PF_joint}, $r$ in~\eqref{eq:PF_r_def}, and $\vartheta$
in~\eqref{eq:PF_theta_def}.

\subsection{Proof of Theorem~\ref{thm:PF_hard}}\label{subsec:PF_proof}

\begin{proof}
     We start by showing that \PFT{} is NP-hard.
    To this end, we show that the construction of the instance $(P_{XY}, r, \vartheta)$
    described above is a polynomial-time reduction from \clique{} to
    \PFT{}. Given an instance $(G,k)$ of \clique{}
    satisfying~\eqref{eq:assumption_on_G}, the reduction constructs the
    extended graph $\G$, the joint distribution $P_{XY}$ defined
    in~\eqref{eq:PF_joint}, and the rational numbers $r$ and $\vartheta$
    defined in~\eqref{eq:PF_r_def} and~\eqref{eq:PF_theta_def}, and
    outputs the instance $(P_{XY}, r, \vartheta)$ of \PFT{}.

    We first show that the reduction runs in time polynomial in
    $n = |V(G)|$. The reduction has to compute $\G$, $P_{XY}$, $r$ and
    $\vartheta$. By property~4 of Lemma~\ref{lemma:gadget}, $\G$ can be
    constructed in polynomial time, and $M = |A(\G)| \le n^4$. We now consider the joint distribution $P_{XY}$ and show that it can be computed in polynomial time. To this end, by~\eqref{eq:PF_Y_x} and~\eqref{eq:PF_joint},  to compute $P_{XY}$ we need the set $\cI(\G)$ of incompatible pairs,
    which determines the alphabet $\cY$, and we need to bound $|\cY|$. We first compute $\cI(\G)$.  There are at most $M^2$ pairs of oriented edges, and,
    for each pair $\{x,x'\}$, deciding whether $x$ and $x'$ are compatible
    requires checking whether the set of their endpoints, which has at
    most four vertices, is a clique of $\G$, that is, checking at most
    $\binom{4}{2} = 6$ adjacencies. Hence, the set $\cI(\G)$ and the numbers $d(x)$
    can be computed in polynomial time. Moreover, by~\eqref{eq:PF_Y_x},
    $\cY$ contains one symbol for each pair in $\cI(\G)$ and $M - d(x)$
    symbols $(x,i)$ for each $x \in A(\G)$. Since each pair of $\cI(\G)$
    contains exactly two oriented edges, we have that
    $\sum_{x \in A(\G)} d(x) = 2|\cI(\G)|$, and hence
    \begin{equation}\label{eq:PF_size_Y}
        |\cY| = |\cI(\G)| + \sum_{x \in A(\G)} \big(M - d(x)\big)
        = |\cI(\G)| + M^2 - 2|\cI(\G)| = M^2 - |\cI(\G)| \le M^2.
    \end{equation}
    Therefore, $P_{XY}$ is a table with at most $M^3$ entries, each equal
    to $1/M^2$ or to $0$ by~\eqref{eq:PF_joint}, which are rational
    numbers with $O(\log n)$ bits, and it can be computed in polynomial
    time. We now consider $r$ and $\vartheta$. By property~1 of
    Lemma~\ref{lemma:PF_channel}, $h_0 = 2|\cI(\G)|/M^2$ is a rational
    number that can be computed exactly, and, by Lemma~\ref{lem:log} in Appendix~\ref{app:log}, a
    rational $\tilde R_0$ with $|\tilde R_0 - R_0| \le \frac{1}{16M}$ can
    be computed in polynomial time. Hence, by~\eqref{eq:PF_r_def}
    and~\eqref{eq:PF_theta_def}, $r$ and $\vartheta$ are rational numbers
    computable in polynomial time. Therefore, the whole reduction runs in polynomial time.

    It remains to check that $r$ and $\vartheta$ are non-negative, as
    required in the definition of \PFT{}. By assumption~\eqref{eq:assumption_on_G} on the graph $G$, we have that $k\geq 4$ and 
    $N \ge B = \binom{k}{2} \ge \binom{4}{2} = 6$, so that
    $R_0 = \log N > 2$. Hence, by the leftmost inequality
    in~\eqref{eq:PF_window}, we get that $r \ge R_0 - \frac{3}{4M} > 1$. Moreover,
    every pair in $\cI(\G)$ consists of two distinct oriented edges, so
    that $|\cI(\G)| \le \binom{M}{2} < \frac{M^2}{2}$. Hence, by property~1 of
    Lemma~\ref{lemma:PF_channel}, it follows that $h_0 = \frac{2|\cI(\G)|}{M^2} < 1$, and therefore,
    by~\eqref{eq:PF_theta_def}, $\vartheta > r - h_0 > 1 - 1 = 0$.

    We now show that the reduction is correct, that is,
    $(P_{XY}, r, \vartheta)$ is a yes-instance of \PFT{} if and only if $(G,k)$ is a yes-instance of \clique{}. We recall that $(P_{XY}, r, \vartheta)$ is a yes-instance of \PFT{} if
    there exists an encoder $T$ of $X$ such that $I(X;T) \ge r$ and
    $I(Y;T) \le \vartheta$, while $(G,k)$ is a yes-instance of \clique{}
    if $\omega(G) \ge k$, that is, if $G$ contains a clique of size $k$. 
    
    Assume first that $(G,k)$ is a yes-instance of \clique{}, that is, $\omega(G) \ge k$.
    By~\eqref{eq:PF_T0} and Lemma~\ref{lemma:PF_yes}, the encoders $T_0$
    and $T_1$ attain the lower bound in~\eqref{eq:PF_lower}, at rates $R_0$ and
    less than $R_0 - \frac{3}{4M}$, respectively. We now show that, by combining the two encoders, we
    obtain an encoder that attains the lower bound in~\eqref{eq:PF_lower} at rate exactly $r$. For
    $\theta \in [0,1]$, let $T_\theta$ be the encoder of
    Lemma~\ref{lemma:mixture} with $T_a = T_0$ and $T_b = T_1$, which
    applies $T_0$ with probability $\theta$ and $T_1$ with probability
    $1-\theta$, and reveals which of the two has been applied.
    By~\eqref{eq:mixture_post}, the posteriors of $T_\theta$ are posteriors
    of $T_0$ or of $T_1$, and are therefore supported on compatible sets.
    Hence, by property~3 of Lemma~\ref{lemma:PF_channel}, it follows that the encoder $T_\theta$ satisfies
    \begin{equation}\label{eq:PF_mixture_equality}
        I(Y;T_\theta) = I(X;T_\theta) - h_0
        \qquad \text{for every } \theta \in [0,1].
    \end{equation}
    We now choose $\theta$ so that $I(X;T_\theta) = r$.
    By the expression
    of the rate of $T_\theta$ in~\eqref{eq:mixture_rate}, 
    $I(X;T_\theta) = \theta\, I(X;T_0) + (1-\theta)\, I(X;T_1)$ is a
    continuous function of $\theta$. Moreover, by Lemma~\ref{lemma:PF_yes},
    by the choice~\eqref{eq:PF_window} of $r$,
    and by the rate of $T_0$ in~\eqref{eq:PF_T0}, we have that
    \begin{equation}\label{eq:PF_r_between}
        I(X;T_1) < R_0 - \frac{3}{4M} \le r
        \le R_0 - \frac{5}{8M} < R_0 = I(X;T_0).
    \end{equation}
    Hence, there exists $\theta \in (0,1)$ such that $I(X;T_\theta) = r$,
    and, for this $\theta$, \eqref{eq:PF_mixture_equality} and
    the definition~\eqref{eq:PF_theta_def} of $\vartheta$ give
    \begin{equation}\label{eq:PF_yes_value}
        I(Y;T_\theta) = r - h_0 = \vartheta - \frac{5}{16M^2} < \vartheta.
    \end{equation}
    Therefore, the encoder $T_\theta$ satisfies both $I(X;T_\theta) \ge r$ and
    $I(Y;T_\theta) \le \vartheta$, and $(P_{XY}, r, \vartheta)$ is a
    yes-instance of \PFT{}. Moreover, by~\eqref{eq:PF_yes_value},
    $\PF(P_{XY},r) \le r - h_0$, while, by Lemma~\ref{lemma:PF_lower},
    $\PF(P_{XY},r) \ge r - h_0$. Hence, it follows that
    \begin{equation}\label{eq:PF_yes_PF}
        \PF(P_{XY},r) = r - h_0 = \vartheta - \frac{5}{16M^2}.
    \end{equation}

    Assume now that $(G,k)$ is a no-instance of \clique{}, that is, $\omega(G) < k$. Let $T$ be an encoder of $X$ with
    $I(X;T) \ge r$. We show that $I(Y;T) > \vartheta$. By Lemma~\ref{lemma:PF_no}, we have that
    \begin{align}\label{eq:PF_no_1}
        I(Y;T) &\ge I(X;T) - h_0 + \frac{R_0 - I(X;T)}{M}\nonumber\\
        &= \Big(1 - \frac1M\Big) I(X;T) + \frac{R_0}{M} - h_0.
    \end{align}
    The right-hand side
    of~\eqref{eq:PF_no_1} is
    increasing in $I(X;T)$. Hence, since $I(X;T) \ge r$, we can
    replace $I(X;T)$ with $r$, and, using the rightmost inequality
    in~\eqref{eq:PF_window} and the definition of $\vartheta$ in~\eqref{eq:PF_theta_def}, we get from~\eqref{eq:PF_no_1} that
    \begin{equation}\label{eq:PF_no_value}
        I(Y;T) \ge r - h_0 + \frac{R_0 - r}{M} \ge r - h_0 + \frac{5}{8M^2}
        = \vartheta + \frac{5}{16M^2} > \vartheta.
    \end{equation}
     Since~\eqref{eq:PF_no_value} holds for any arbitrary encoder $T$ with $I(X;T) \ge r$, it follows that no encoder
    satisfies both $I(X;T) \ge r$ and $I(Y;T) \le \vartheta$, and hence
    $(P_{XY}, r, \vartheta)$ is a no-instance of \PFT{}. Moreover,
    taking the infimum over all encoders with $I(X;T) \ge r$
    in~\eqref{eq:PF_no_value} yields
    \begin{equation}\label{eq:PF_no_PF}
        \PF(P_{XY},r) \ge \vartheta + \frac{5}{16M^2}.
    \end{equation}
    This proves that the reduction is correct, and hence that \PFT{} is
    NP-hard.

    It remains to prove the inapproximability result.
    By~\eqref{eq:PF_yes_PF} and~\eqref{eq:PF_no_PF}, $\PF(P_{XY},r)$ lies
    below $\vartheta - \frac{5}{16M^2}$ if $\omega(G) \ge k$, and above
    $\vartheta + \frac{5}{16M^2}$ otherwise. Suppose that a polynomial-time
    algorithm outputs, given $P_{XY}$ and $r$, a value $v$ with
    $|v - \PF(P_{XY},r)| \le \frac{1}{4|\cX|^2} = \frac{4}{16M^2}$. If
    $\omega(G) \ge k$, then, by~\eqref{eq:PF_yes_PF},
    $v \le \PF(P_{XY},r) + \frac{4}{16M^2} = \vartheta - \frac{1}{16M^2} < \vartheta$,
    while, if $\omega(G) < k$, then, by~\eqref{eq:PF_no_PF},
    $v \ge \PF(P_{XY},r) - \frac{4}{16M^2} \ge \vartheta + \frac{1}{16M^2} > \vartheta$.
    Thus, comparing $v$ with $\vartheta$ would decide \clique{} in
    polynomial time, which is impossible unless
    $\mathrm{P} = \mathrm{NP}$.
\end{proof}

\subsection{The Lagrangian formulation}\label{subsec:PF_lagrangian}

We now show that \PFL$_\beta$ is NP-hard, and hard to approximate, for
every fixed rational $\beta \in (0,1)$. The proof combines the
constructions in the proofs of Theorem~\ref{thm:PF_hard} and
Proposition~\ref{prop:other_beta}. Specifically, we first reduce \clique{} to
\PFL$_{1/2}$, starting from the same
instance of Theorem~\ref{thm:PF_hard}, and then we reduce \PFL$_{1/2}$ to
\PFL$_\beta$, as in Proposition~\ref{prop:other_beta}.

\begin{proposition}\label{prop:PF_lagrangian}
    Let $\beta \in \mathbb{Q} \cap (0,1)$ be fixed, and let
    $\kappa_\beta = 2\min\{\beta, 1-\beta\}$. Then \PFL$_\beta$ is
    NP-hard. Moreover, unless $\mathrm{P} = \mathrm{NP}$, there is no
    polynomial-time algorithm that, given a rational joint distribution
    $P_{XY}$ on $\cX \times \cY$, outputs a value $v$ such that
    \begin{equation*}
        \big| v - \Lambda_\beta(P_{XY}) \big| \le \frac{\kappa_\beta}{8|\cX|^2}.
    \end{equation*}
\end{proposition}

\begin{proof}
    We first prove that \PFL$_\beta$ is NP-hard for every
    $\beta \in \mathbb{Q} \cap (0,1)$. To this end, we give a polynomial-time reduction from \clique{} to \PFL$_{1/2}$,
    and then a polynomial-time reduction from \PFL$_{1/2}$ to
    \PFL$_\beta$.

    We start with the reduction from \clique{} to \PFL$_{1/2}$. Let
    $(G,k)$ be an instance of \clique{} satisfying~\eqref{eq:assumption_on_G},
    let $P_{XY}$ be the joint distribution defined in~\eqref{eq:PF_joint},
    and let
    \begin{equation}\label{eq:PF_lag_params}
        \beta_0 = 1 - \frac1M, \qquad
        \ell_0 = \frac{R_0}{M} - h_0, \qquad
        a = \frac{1}{2\beta_0} = \frac{M}{2(M-1)}.
    \end{equation}
    Note that $\frac12 < a \le 1$, since $M \ge 2$. We define the joint
    distribution $P_{XY'}$ on $\cX \times (\cY \cup \{\bot\})$, where
    $\bot$ is a new symbol, by
    \begin{equation}\label{eq:PF_lag_joint}
        P_{XY'}(x,y) = a\, P_{XY}(x,y), \quad \forall y \in \cY,
        \qquad\text{and}\quad
        P_{XY'}(x,\bot) = (1-a)\, P_X(x),
    \end{equation}
    that is, $Y'$ is equal to $Y$ with probability $a$, and to $\bot$
    otherwise, independently of $(X,Y)$. Moreover, we let $\tau$ be a
    rational number such that
    \begin{equation}\label{eq:PF_lag_tau}
        \Big| \tau - \Big( a\ell_0 - \frac{6}{32M^2} \Big) \Big| \le \frac{1}{32M^2}.
    \end{equation}
    The reduction maps $(G,k)$ to the instance $(P_{XY'}, \tau)$ of
    \PFL$_{1/2}$. We first show that the instance $(P_{XY'}, \tau)$ can be computed in polynomial time.
    Indeed, the joint distribution $P_{XY}$ can be computed in polynomial time, as shown in the
    proof of Theorem~\ref{thm:PF_hard}, and, since $a$ is rational, it follows by~\eqref{eq:PF_lag_joint} that 
    $P_{XY'}$ can also be computed in polynomial time. Moreover,
    $a\ell_0 = \frac{a}{M}\log N - a h_0$, where $a$ and $h_0$ are
    rational. Hence, by Lemma~\ref{lem:log}, $\tau$ can be computed in
    polynomial time. Therefore, the whole instance $(P_{XY'}, \tau)$ can be computed in polynomial time.

    We now show that the reduction is correct.  By
    Lemmas~\ref{lemma:PF_no} and~\ref{lemma:PF_yes}, the Lagrangian of
    $P_{XY}$ with multiplier $\beta_0 = 1 - \frac1M$ already separates the
    two cases of \clique{}. However, $\beta_0$ depends on the instance,
    while the multiplier of \PFL$_{1/2}$ is fixed. For this reason, we  relate the
    Lagrangian of $P_{XY'}$ with multiplier $\frac12$ to the Lagrangian of
    $P_{XY}$ with multiplier $\beta_0$. Since $X$ is unchanged, the
    encoders of $X$ for $P_{XY}$ and for $P_{XY'}$ coincide. Let $T$ be an
    encoder of $X$, and let $E$ be the indicator of the event
    $\{Y' \ne \bot\}$, so that $E = 1$ with probability $a$. Since $E$ is
    a function of $Y'$, by the chain rule, we have that
    \begin{equation}\label{eq:PF_lag_chain}
        I(Y';T) = I(Y',E;T) = I(E;T) + I(Y';T \mid E).
    \end{equation}
    By construction, $E$ is independent of $(X,Y)$, and $T$ depends only
    on $X$. Thus, $E$ is independent of $(X,Y,T)$, and it follows that
    $I(E;T) = 0$, and that, for each value of $E$, the distribution of
    $(Y,T)$ is unchanged. Moreover, if $E = 1$, then $Y' = Y$, so that
    $I(Y';T \mid E = 1) = I(Y;T)$, while, if $E = 0$, then $Y' = \bot$ is
    constant, so that $I(Y';T \mid E = 0) = 0$. Substituting these facts
    into~\eqref{eq:PF_lag_chain}, we obtain that
    \begin{equation}\label{eq:PF_lag_Yprime}
        I(Y';T) = a\, I(Y;T).
    \end{equation}
    Finally, since $a\beta_0 = \frac12$ by~\eqref{eq:PF_lag_params},
    subtracting $\frac12 I(X;T)$ from both sides of~\eqref{eq:PF_lag_Yprime}
    gives
    \begin{equation}\label{eq:PF_lag_half}
        I(Y';T) - \tfrac12 I(X;T) = a\, I(Y;T) - a\beta_0\, I(X;T)
        = a\big[ I(Y;T) - \beta_0\, I(X;T) \big].
    \end{equation}
    We now use the equivalence~\eqref{eq:PF_lag_half} to show that $G$ contains a
    $k$-clique if and only if there exists an encoder $T$ of $X$ such that
    \begin{equation}\label{eq:PF_lag_goal}
        I(Y';T) - \tfrac12 I(X;T) \le \tau.
    \end{equation}

    Assume first that $(G,k)$ is a no-instance of \clique{}, that is, $\omega(G) < k$. By Lemma~\ref{lemma:PF_no},
    and since
    $\beta_0 = 1 - \frac1M$ and $\ell_0 = \frac{R_0}{M} - h_0$
    by~\eqref{eq:PF_lag_params}, every encoder $T$ of $X$ satisfies
    \begin{equation*}
        I(Y;T) \ge I(X;T) - h_0 + \frac{R_0 - I(X;T)}{M} = \beta_0\, I(X;T) + \ell_0,
    \end{equation*}
    that is, $I(Y;T) - \beta_0 I(X;T) \ge \ell_0$. Hence,
    by the equivalence~\eqref{eq:PF_lag_half} and the definition of $\tau$ in~\eqref{eq:PF_lag_tau}, every encoder $T$
    of $X$ satisfies
    \begin{equation}\label{eq:PF_lag_no}
        I(Y';T) - \tfrac12 I(X;T) \ge a\ell_0 \ge \tau + \frac{5}{32M^2}.
    \end{equation}
    Thus, no encoder satisfies~\eqref{eq:PF_lag_goal}, and
    $(P_{XY'}, \tau)$ is a no-instance of \PFL$_{1/2}$.

    Assume now that $(G,k)$ is a yes-instance of \clique{}, that is, $\omega(G) \ge k$. By Lemma~\ref{lemma:PF_yes}, the
    encoder $T_1$ satisfies $I(Y;T_1) = I(X;T_1) - h_0$ and
    $I(X;T_1) < R_0 - \frac{3}{4M}$. Hence, substituting
    $I(Y;T_1) = I(X;T_1) - h_0$ and $\beta_0 = 1 - \frac1M$, we get that
    \begin{equation}\label{eq:PF_lag_yes_beta0}
        I(Y;T_1) - \beta_0\, I(X;T_1) = \frac{I(X;T_1)}{M} - h_0
        < \frac{R_0}{M} - \frac{3}{4M^2} - h_0 = \ell_0 - \frac{3}{4M^2},
    \end{equation}
    where the inequality follows from $I(X;T_1) < R_0 - \frac{3}{4M}$, and
    the last equality from the definition of $\ell_0$ in~\eqref{eq:PF_lag_params}. Multiplying both sides of~\eqref{eq:PF_lag_yes_beta0} by $a > 0$, and using the identity~\eqref{eq:PF_lag_half}
    and $a > \frac12$, we get that
    \begin{equation}\label{eq:PF_lag_yes_a}
        I(Y';T_1) - \tfrac12 I(X;T_1) < a\ell_0 - \frac{3a}{4M^2}
        < a\ell_0 - \frac{3}{8M^2} = a\ell_0 - \frac{12}{32M^2}.
    \end{equation}
     Moreover, by the definition~\eqref{eq:PF_lag_tau} of $\tau$, we have
    that $a\ell_0 \le \tau + \frac{6}{32M^2} + \frac{1}{32M^2} = \tau + \frac{7}{32M^2}$.
    Substituting this into~\eqref{eq:PF_lag_yes_a}, we obtain that
    \begin{equation}\label{eq:PF_lag_yes}
        I(Y';T_1) - \tfrac12 I(X;T_1) < \tau + \frac{7}{32M^2} - \frac{12}{32M^2}
        = \tau - \frac{5}{32M^2},
    \end{equation}
    Thus, the encoder $T_1$ satisfies~\eqref{eq:PF_lag_goal}, and $(P_{XY'}, \tau)$ is a
    yes-instance of \PFL$_{1/2}$. This proves that the reduction is
    correct, and hence that \PFL$_{1/2}$ is NP-hard.

    We now prove that $\Lambda_{1/2}$ is hard to approximate.
    By~\eqref{eq:PF_lag_no} and~\eqref{eq:PF_lag_yes}, we have that
    $\Lambda_{1/2}(P_{XY'}) \ge \tau + \frac{5}{32M^2}$ if $\omega(G) < k$,
    and that $\Lambda_{1/2}(P_{XY'}) < \tau - \frac{5}{32M^2}$ if
    $\omega(G) \ge k$. Suppose that a polynomial-time algorithm outputs,
    given $P_{XY'}$, a value $v$ with
    $|v - \Lambda_{1/2}(P_{XY'})| \le \frac{1}{8|\cX|^2}$. Since $P_{XY'}$
    has the same input alphabet $\cX$ as $P_{XY}$, we have that
    $|\cX| = M$, and hence $\frac{1}{8|\cX|^2} = \frac{4}{32M^2}$.
    Therefore, if $\omega(G) < k$, then
    $v \ge \Lambda_{1/2}(P_{XY'}) - \frac{4}{32M^2} \ge \tau + \frac{1}{32M^2} > \tau$,
    while, if $\omega(G) \ge k$, then
    $v \le \Lambda_{1/2}(P_{XY'}) + \frac{4}{32M^2} < \tau - \frac{1}{32M^2} < \tau$.
    Thus, running the reduction, then the algorithm, and comparing $v$ with $\tau$ would decide \clique{} in polynomial
    time, which is impossible unless $\mathrm{P} = \mathrm{NP}$.

    It remains to give the reduction from \PFL$_{1/2}$ to \PFL$_\beta$.
    Let $(P_{XY'}, \tau)$ be an instance of \PFL$_{1/2}$, and let
    $\tilde P_{XZ}$ be the distribution obtained from $P_{XY'}$ by the
    construction~\eqref{eq:Z_def} in the proof of
    Proposition~\ref{prop:other_beta}. The reduction maps
    $(P_{XY'}, \tau)$ to the instance $(\tilde P_{XZ}, \kappa_\beta\tau)$
    of \PFL$_\beta$, which, as shown in that proof, can be computed in
    polynomial time. Moreover, by~\eqref{eq:L_beta_identity}, every encoder $T$ of
    $X$ satisfies
    \begin{equation}\label{eq:PF_lag_beta}
        I(Z;T) - \beta I(X;T) = \kappa_\beta \Big( I(Y';T) - \tfrac12 I(X;T) \Big).
    \end{equation}
    Since $\kappa_\beta > 0$, by~\eqref{eq:PF_lag_beta} an encoder
    satisfies $I(Y';T) - \frac12 I(X;T) \le \tau$ if and only if it
    satisfies $I(Z;T) - \beta I(X;T) \le \kappa_\beta\tau$. Hence, the
    reduction is correct, and \PFL$_\beta$ is NP-hard. Finally, taking the
    infimum over all encoders in~\eqref{eq:PF_lag_beta}, we get that
    $\Lambda_\beta(\tilde P_{XZ}) = \kappa_\beta \Lambda_{1/2}(P_{XY'})$.
    Since $\tilde P_{XZ}$ has the same input alphabet $\cX$ as $P_{XY'}$,
    an approximation of $\Lambda_\beta$ within
    $\frac{\kappa_\beta}{8|\cX|^2}$, divided by $\kappa_\beta$, would give
    an approximation of $\Lambda_{1/2}$ within $\frac{1}{8|\cX|^2}$,
    which, as shown above, is impossible unless
    $\mathrm{P} = \mathrm{NP}$.
\end{proof}

We observe that the argument of Remark~\ref{rem:real_beta} for the
information bottleneck applies also to the privacy funnel. Therefore,
the result of Proposition~\ref{prop:PF_lagrangian} extends to every fixed
polynomial-time computable real $\beta \in (0,1)$. Finally, note that,
for $\beta \ge 1$, the problem is trivial since the optimum is attained by $T=X$. Indeed, by
Lemma~\ref{lemma:PF_lower}, every encoder $T$ of $X$ satisfies
\begin{equation*}
    I(Y;T) - \beta I(X;T) \ge (1-\beta)\, I(X;T) - H(X \mid Y)
    \ge (1-\beta)\, H(X) - H(X \mid Y) = I(X;Y) - \beta H(X),
\end{equation*}
where the second inequality follows from $\beta \ge 1$ and
$I(X;T) \le H(X)$, and equality holds for $T = X$. Hence,
$\Lambda_\beta(P_{XY}) = I(X;Y) - \beta H(X)$ for every $\beta \ge 1$.

\section{Conclusions}\label{sec:conclusions}

We have established the NP-hardness of the threshold decision problems
for both the constrained and the Lagrangian formulations of the
information bottleneck.
Theorems~\ref{thm:IBL_hard}
and~\ref{thm:IBT_hard} also rule out polynomial-time approximation within
additive errors of order $|\cX|^{-3/2}$, unless
$\mathrm{P} = \mathrm{NP}$, and, by padding the instances of \clique{}
with isolated edges, the exponent $3/2$ can be lowered to $1 + \alpha$, for
every fixed $\alpha > 0$ (Theorem~\ref{thm:padding}). 

By adapting our hardness framework for the information bottleneck, we
also showed that it is NP-hard both to compute the optimal privacy
funnel~\cite{MSFM14} exactly and to approximate it within an additive
error of order $|\cX|^{-2}$ (Theorem~\ref{thm:PF_hard}). The same holds
for its Lagrangian formulation, for every fixed polynomial-time
computable multiplier $\beta \in (0,1)$
(Proposition~\ref{prop:PF_lagrangian}), while, for $\beta \ge 1$,
$\Lambda_\beta(P_{XY}) = I(X;Y) - \beta H(X)$.

An open question is whether approximating the information bottleneck and
the privacy funnel within a larger additive error remains NP-hard.
Thus, whether
hardness holds within errors of order $|\cX|^{-1}$, or even within a
constant additive error, remains open.

Given the many applications of the information bottleneck framework,
perhaps the most important open
problem is the design of efficient approximation algorithms with provable
guarantees.
Another  direction is to identify structural conditions on the joint
distribution that permit efficient exact computation or approximation.

\section*{AI statement}
The authors used OpenAI's ChatGPT Astra and Anthropic's Claude Opus 5.5 throughout the preparation of this manuscript for language editing, expository suggestions, mathematical consistency checks, and literature and reference checks. ChatGPT drew the authors' attention to the relevance of the Motzkin--Straus lemma~\cite{MS65} and the entropic form of Turán's theorem~\cite{CY26} to this work. 
The overall problem formulation
and the overall structure of the paper are the authors’ own.
The authors reviewed and edited the AI-assisted material and take full responsibility for the content of the paper.

\appendices

\section{Approximation of logarithms}\label{app:log}

The reductions of Theorems~\ref{thm:IBL_hard}, \ref{thm:IBT_hard}
and~\ref{thm:PF_hard}, and of Proposition~\ref{prop:PF_lagrangian}, must
compute rational approximations of finitely many logarithmic
expressions. Exact comparison of such expressions is a  delicate
matter (e.g., see~\cite{ESY14}). However, the 
reductions of Theorems~\ref{thm:IBL_hard}, \ref{thm:IBT_hard}
and~\ref{thm:PF_hard}, and of Proposition~\ref{prop:PF_lagrangian}  only need the following standard
approximation result, which follows from classical algorithms for the
multiple-precision evaluation of the logarithm (see~\cite{Brent76}).

\begin{lemma}\label{lem:log}
    There is an algorithm that, given rational numbers
    $\zeta_0, \dots, \zeta_\ell$, positive rational numbers
    $z_1, \dots, z_\ell$, and a rational number $\gamma > 0$, outputs a
    rational number $\tilde F$ such that
    \begin{equation*}
        \Big| \tilde F - \Big( \zeta_0 + \sum_{i=1}^{\ell} \zeta_i \log z_i \Big) \Big| \le \gamma,
    \end{equation*}
    in time polynomial in the total bit length of the input.
\end{lemma}


\begin{thebibliography}{99}

\bibitem{Abu}
Y.~S. Abu-Mostafa, 
Essential average mutual information, in: \textit{Open Problems in
Communication and Computation}, T.~M. Cover and B. Gopinath (Eds.), 
Springer--Verlag, pp. 75, 1987.

\bibitem{Agmon2023}
S. Agmon, The information bottleneck’s ordinary differential equation: First-order root tracking for the information bottleneck, \textit{Entropy}, vol.~{\bf25}, no.~10, 2023.

\bibitem{AC20}
S. Asoodeh and F.~P. Calmon, Bottleneck problems: An information and
estimation-theoretic view, \emph{Entropy}, vol.~{\bf22}, no.~11, 2020.

\bibitem{BAC23}
E.~Benger, S.~Asoodeh and J.~Chen, The cardinality bound on the
information bottleneck representations is tight, in \emph{Proc.\ IEEE
International Symposium on Information Theory (ISIT)}, Taipei, Taiwan, pp.~1478--1483,
2023. 

\bibitem{Brent76}
R.~P. Brent, Fast multiple-precision evaluation of elementary functions,
\emph{Journal of the ACM}, vol.~{\bf23}, pp.~242--251, 1976.

\bibitem{CY26}
T.-W.~Chao and H.-H.~H.~Yu, When entropy meets Tur\'an: new proofs and
hypergraph Tur\'an results, \emph{J.\ London Math.\ Soc.}, vol.~{\bf113},
no.~3, 2026.


\bibitem{CS89}
N.~Chayat and S.~Shamai, Extension of an entropy property for binary
input memoryless symmetric channels, \emph{IEEE Trans.\ Inf.\ Theory},
vol.~{\bf35}, no.~5, pp.~1077--1079, 1989.

\bibitem{CGTW05}
G.~Chechik, A.~Globerson, N.~Tishby and Y.~Weiss, Information bottleneck
for Gaussian variables, \emph{Journal of Machine Learning Research},
vol.~{\bf6}, pp.~165--188, 2005.

\bibitem{DWS22}
M.~Dikshtein, N.~Weinberger and S.~Shamai, On information bottleneck for
Gaussian processes, in \emph{Proc.\ IEEE Information Theory Workshop
(ITW)}, Mumbai, India, pp.~546--551, 2022.

\bibitem{ESY14}
K.~Etessami, A.~Stewart and M.~Yannakakis, A note on the complexity of
comparing succinctly represented integers, with an application to maximum
probability parsing, \emph{ACM Trans.\ Comput.\ Theory}, vol.~{\bf6}, no.~2, 2014.

\bibitem{FriedmanMST01} N. Friedman, O. Mosenzon, N. Slonim and N. Tishby,
Multivariate information bottleneck, in \textit{Proc. 17th Conference on Uncertainty
in Artificial Intelligence (UAI)}, Seattle, WA, USA, pp.~152--161, 2001.


\bibitem{GJ79}
M.~R.~Garey and D.~S.~Johnson, \emph{Computers and Intractability: A
Guide to the Theory of NP-Completeness}, W.~H.~Freeman, 1979.

\bibitem{GP20}
Z.~Goldfeld and Y.~Polyanskiy, The information bottleneck problem and its
applications in machine learning, \emph{IEEE J.\ Sel.\ Areas Inf.\
Theory}, vol.~{\bf1}, no.~1, pp.~19--38, 2020.

\bibitem{HLYY24}
S.~Hu, Z.~Lou, X.~Yan and Y.~Ye, A survey on information bottleneck,
\emph{IEEE Trans.\ Pattern Anal.\ Mach.\ Intell.}, vol.~{\bf46}, no.~8,
pp.~5325--5344, 2024. 

\bibitem{HuangElGamal2021}
T.-H. Huang and A. El Gamal, A provably convergent information bottleneck solution via ADMM, in \textit{Proceedings of the IEEE International Symposium on Information Theory (ISIT)}, Melbourne, Australia, pp.~43–-48, 2021.


\bibitem{Karp72}
R.~M.~Karp, Reducibility among combinatorial problems, in
\emph{Complexity of Computer Computations} (R.~E.~Miller and
J.~W.~Thatcher, eds.), Plenum Press, New York, pp.~85--103, 1972.

\bibitem{KTV19}
A.~Kolchinsky, B.~D.~Tracey and S.~Van Kuyk, Caveats for information
bottleneck in deterministic scenarios, in \emph{Proc.\ International
Conference on Learning Representations (ICLR)}, New Orleans, LA, USA, 2019.


\bibitem{Lew}
J. Lewandowsky and G.  Bauch, Eds.,  \textit{Theory and Application of the Information Bottleneck Method}. Basel, Switzerland: MDPI, 2024.

\bibitem{Liu}
H. Liu and Y. Guo, Survey of information bottleneck frameworks in multimodal representation learning, 
\textit{Journal of Frontiers of Computer Science and Technology}, vol.~{\bf20}, no.~9, pp.~2454-–2468, 2026.

\bibitem{MSFM14}
A.~Makhdoumi, S.~Salamatian, N.~Fawaz and M.~M\'edard, From the
information bottleneck to the privacy funnel, in \emph{Proc.\ IEEE
Information Theory Workshop (ITW)}, Hobart, Australia, pp.~501--505,
2014. 

\bibitem{MS65}
T.~S.~Motzkin and E.~G.~Straus, Maxima for graphs and a new proof of a
theorem of Tur\'an, \emph{Canad.\ J.\ Math.}, vol.~{\bf17}, pp.~533--540,
1965.

\bibitem{Pia}
P. Piantanida and  L. Rey Vega, Information bottleneck and representation learning, in:
\textit{Information-Theoretic Methods in Data Science},
 M.~R.~D. Rodrigues and Y.~C. Eldar (Eds.), Cambridge University Press, pp. 330–-358, 2021.

\bibitem{SV16}
I.~Sason and S.~Verd\'u, $f$-divergence inequalities, \emph{IEEE Trans.\
Inf.\ Theory}, vol.~{\bf62}, no.~11, pp.~5973--6006, 2016.

\bibitem{ZL24}
R.~Shwartz-Ziv and Y.~LeCun, To compress or not to
compress---self-supervised learning and information theory: A review,
\emph{Entropy}, vol.~{\bf26}, no.~3, 2024.

\bibitem{StrouseS17} D.~J. Strouse and D.~J. Schwab, The deterministic
information bottleneck, \textit{Neural Computation}, vol.~{\bf29}, no.~6, pp.~1611--1630, 2017.

\bibitem{SSZ05}
I.~Sutskover, S.~Shamai and J.~Ziv, Extremes of information combining,
\emph{IEEE Trans.\ Inf.\ Theory}, vol.~{\bf51}, no.~4, pp.~1313--1325, 2005.

\bibitem{TPB99}
N.~Tishby, F.~C.~Pereira and W.~Bialek, The information bottleneck
method, in \emph{Proc.\ 37th Annual Allerton Conference on Communication,
Control, and Computing}, Monticello, IL, USA, pp.~368--377, 1999.

\bibitem{WW75}
H.~S. Witsenhausen and A.~D. Wyner, A conditional entropy bound for a pair of
discrete random variables, \emph{IEEE Trans.\ Inf.\ Theory}, vol.~{\bf21},
pp.~493--501, 1975.

\bibitem{WZ73b}
A.~D. Wyner, A theorem on the entropy of certain binary sequences and
applications: Part~II, \emph{IEEE Trans.\ Inf.\ Theory}, vol.~{\bf19},
pp.~772--777, 1973.

\bibitem{WZ73a}
A.~D. Wyner and J.~Ziv, A theorem on the entropy of certain binary sequences
and applications: Part~I, \emph{IEEE Trans.\ Inf.\ Theory}, vol.~{\bf19},
pp.~769--772, 1973.

\bibitem{ZEAS20}
A.~Zaidi, I.~Estella-Aguerri and S.~Shamai (Shitz), On the information bottleneck
problems: Models, connections, applications and information theoretic
views, \emph{Entropy}, vol.~{\bf22}, no.~2, 2020.
\end{thebibliography}
\end{document}